\documentclass[journal]{IEEEtran}
\usepackage{amsthm}
\usepackage{eucal}
\usepackage{dsfont}
\usepackage{times}
\usepackage{mathptmx}
\usepackage{tikz}
\usepackage{verbatim}
\usetikzlibrary{shapes,arrows}
\usepackage{subfigure}
\usepackage{enumerate}
\newtheorem{lemm}{Lemma}
\newtheorem{defin}{Definition}
\newtheorem{assump}{Assumption}

\newtheorem{theor}{Theorem}

\newtheorem{rem}{Remark}
\newtheorem{probl}{Problem}
\newtheorem{examp}{Example}
\usepackage{amssymb,amsfonts}
\usepackage{amsmath} 
\usepackage{algorithmic}
\usepackage{epstopdf}
\usepackage{graphicx}
\usepackage{textcomp}
\usepackage{xcolor}
\usepackage[numbers,sort&compress]{natbib}
\def\BibTeX{{\rm B\kern-.05em{\sc i\kern-.025em b}\kern-.08em
    T\kern-.1667em\lower.7ex\hbox{E}\kern-.125emX}}

\begin{document}
\title{Order-Adaptive Distributed Integral Control
\thanks{Fei~Chen is with the College of Artificial Intelligence, Nankai University, Tianjin 300350, China.} \thanks{Correspondence should be addressed to Fei Chen (e-mail: fchen@nankai.edu.cn).}
}

\author{Fei~Chen,~\IEEEmembership{Senior Member,~IEEE}}

\maketitle

\begin{abstract}
We address a structural tradeoff in distributed dynamic coordination: when
the target complexity is unknown, a low controller order saves states but may
leave a persistent tracking error, whereas a high order improves tracking but
may burden every agent with unnecessary dynamics. To remove this choice
without resorting to computationally more involved nonlinear feedback or
chattering-prone nonsmooth feedback, we develop an
order-adaptive distributed integral controller (OADIC). Specifically, we
start with proportional feedback and add integral states only when locally
measurable relative errors show that the current order is inadequate.
Meanwhile, we organize the candidate controllers in a nested form, thereby
preserving the existing states and gains and maintaining continuous control
inputs during order transitions. To provide a theoretical basis for this design,
we first characterize the consistency of prescribed relative displacements
on the augmented agent--target graph. Next, we construct gains that stabilize
all admissible fixed-order subsystems and establish a uniform input-to-state
stability bound for the variable-dimension closed loop. Furthermore, we prove
that OADIC rejects every insufficient order after finitely many
decision intervals and explicitly bound the rejection time of the critical
order.
\end{abstract}


\section{Introduction}

How much internal dynamics should a distributed controller carry? The
internal model principle offers a partial answer: exact regulation
requires the controller to reproduce the dynamics that generate the signal to
be tracked or rejected \cite{francis1976internal}. Yet the principle does not
specify how to choose the internal-model order when the complexity of that
signal is unknown. This uncertainty therefore creates a structural tradeoff: an
order chosen too low leaves unmodeled dynamics, whereas the highest
anticipated order forces every agent to carry unnecessary integral states
whenever the actual motion is simpler. This tradeoff is particularly
restrictive in large-scale distributed networks because the controller
dynamics are replicated at every node: even one superfluous integral state
per agent creates network-wide storage, computation, and initialization
overheads that grow with the number of agents, whereas
underestimating the order leaves a persistent coordination error throughout
the network. Moreover, because each agent has access only to local relative
information, distinguishing an inadequate internal model from a slowly
decaying network transient is itself highly nontrivial.

The scalability issue mentioned above has motivated extensive efforts to reduce the
global information and model knowledge required by distributed controllers.
A common approach nevertheless assigns each agent a copy or an estimate of a
prescribed reference-signal model. This idea provides the basis for output
synchronization conditions \cite{wieland2011internal} and for distributed
observer--regulator architectures that coordinate heterogeneous agents
\cite{su2012cooperative}. More recently, fully distributed adaptive protocols
have removed the need for global network information while retaining
consensus guarantees \cite{li2013distributedadaptive}. Such advances make
distributed coordination more scalable and less dependent on global network
information. However, they still specify the reference-signal model and its
dimension before operation.

Rather than reconstructing an explicit reference-signal model, another line
of research embeds the target dynamics through integral or filter states in
the relative-feedback loop. Within this framework, distributed integral
action improves regulation and disturbance rejection
\cite{andreasson2014distributed}, whereas dynamic average consensus and the
closely related distributed average tracking methods enable networks to
follow time-varying signals
\cite{freeman2006dynamic,chen2020distributedaverage}. This interpretation is especially transparent for
polynomial targets: each additional integrator incorporates one more constant
derivative into the controller and thereby leads naturally to fixed
$PI^n$-type tracking schemes \cite{cheng2016pin}. Alongside these dynamic
compensation methods, recent adaptive and data-driven approaches have further
reduced the need for accurate agent models. In particular, output-feedback
adaptation estimates unknown agent parameters online \cite{guo2025mrac},
whereas data-driven synchronization constructs controllers directly from
measured trajectories \cite{fattore2026data}. Although these approaches
reduce prior model knowledge in different ways, they retain a controller
architecture whose state dimension is prescribed in advance. Consequently,
none of them determines during operation whether the observed tracking
performance calls for an additional integral state.

In parallel, nonlinear and nonsmooth control methods seek stronger transient
guarantees by reshaping the feedback law. Fixed-time control, for example,
provides a settling-time bound independent of the initial condition
\cite{polyakov2012fixed}, and its distributed extensions carry this property
over to networks of integrators \cite{parsegov2013fixed,chen2015boundedaccelerations}. Building on this
foundation, recent designs combine distributed observers with sliding-mode
feedback to achieve fixed-time formation tracking
\cite{zhou2025fixedtime}, while nonlinear output-feedback schemes enforce
prescribed tracking performance for uncertain high-order networks
\cite{katsoukis2026output}. These stronger performance guarantees, however,
come with a more involved controller structure. In particular, nonlinear
gains, fractional-power or sliding-mode terms, and additional observer states
must often be implemented at every agent, so their aggregate cost increases
with the network size. Moreover, nonsmooth terms may cause chattering and
require special solution definitions. Consequently, this line of research
improves tracking by increasing the complexity of the feedback law, but it
still leaves the internal-model dimension as an a priori design choice.

The above results leave a clear gap between low- and high-order linear
controllers. A low-order controller is inexpensive but may
leave a nonzero tracking error, whereas a high-order controller can track
richer motions but carries more states at every agent. A fixed-order design
must choose one of these two outcomes before operation and cannot move from
one to the other when the target motion changes. As a related alternative,
nonlinear or nonsmooth feedback can improve transient performance or
robustness, but it generally retains an a priori controller dimension and
introduces additional implementation complexity. Thus, the methods reviewed
above do not provide online order adjustment within a purely linear
distributed architecture.

To close this gap, we develop an order-adaptive distributed integral
controller (OADIC) for multi-agent networks with consistent
relative-displacement constraints. Rather than selecting a high fixed order
in advance or introducing nonlinear feedback, OADIC begins with the simplest
linear controller and increases its order only when the observed tracking
progress is insufficient. Importantly, at every active order, the control
input remains a purely linear combination of the coordination error and a
nested sequence of integral states. To determine whether this order remains
adequate, the agents then construct a distributed monitor from their local
relative errors. Accordingly, at each decision instant, OADIC retains the
current order if the monitor decreases sufficiently; otherwise, it activates
one additional integral state. 

The main contributions are summarized as follows.
\begin{enumerate}

    \item We develop OADIC as a linear distributed control framework
    that achieves dynamic coordination while adapting controller complexity
    to the observed tracking performance. To determine whether more
    controller states are needed, OADIC constructs a distributed monitor from
    local relative errors. If the monitor indicates that the current order is
    inadequate, OADIC activates one additional integral state. Because the
    controllers are nested, this activation preserves the existing states and
    gains while maintaining continuity of the control input.

    \item We give a constructive nested gain condition that simultaneously
    stabilizes every admissible fixed-order subsystem. By accounting for both
    continuous target forcing and the reset maps, we establish a uniform
    input-to-state stability bound for the complete variable-dimension
    dynamics. We also derive convergence estimates connecting the
    threshold schedule with the tracking performance of the final
    unsaturated order.

    \item 
    We prove that every insufficient order is rejected after finitely many decision intervals and
    obtain an explicit rejection-time bound for the critical order. This
    explains why OADIC can identify missing integral action even when the
    active fixed-order subsystem is itself stable.
\end{enumerate}

The remainder of the paper is organized as follows. Section~II introduces
the notation and auxiliary results. Section~III formulates the dynamic
coordination problem. Section~IV characterizes consistency of the prescribed
relative displacements. Section~V presents OADIC, its distributed monitor,
the switching rule, and the nested gain design. Section~VI establishes
stability and tracking properties, and Section~VII analyzes finite-time
detection of insufficient orders. Section~VIII provides a numerical
comparison with fixed-order controllers, and Section~IX concludes the paper.

\section{Preliminaries}

\label{sec:prel}

\subsection{Notation}
We use $\mathbb R$ and $\mathbb Z^+$ to denote the sets of real numbers and
positive integers, respectively, and we write $\mathbb R^n$ and
$\mathbb R^{n\times m}$ for the spaces of real $n$-vectors and $n\times m$
matrices. We use the superscript $T$ for transposition. For a vector $x$, we
use $\|x\|$ for the Euclidean norm; for a matrix $M$, we use the same notation
for the induced Euclidean norm.
We use $\mathbf 1_n$, $\mathbf 0_n$, $I_n$, and $0_{n\times m}$ for the
all-ones vector, the zero vector, the identity matrix, and the zero matrix of
the indicated dimensions. We omit a subscript when its dimension is clear
from the context. For a switching instant $t_s$ and a signal $\xi(t)$ with one-sided limits, we
write
$\xi(t_s^-)\triangleq\lim_{t\uparrow t_s}\xi(t)$ and
$\xi(t_s^+)\triangleq\lim_{t\downarrow t_s}\xi(t)$. 

For scalars $a_1,\ldots,a_n$, we use
$\operatorname{diag}(a_1,\ldots,a_n)$ for the diagonal matrix with these
entries and $\otimes$ for the Kronecker product. We define
$[a]_+\triangleq\max\{a,0\}$ and use $\lceil a\rceil$ for the smallest integer
not less than $a$. For a sufficiently differentiable signal $x(t)$, we write
$x^{(r)}(t)$ for its $r$th derivative and set $\dot x=x^{(1)}$. For a real symmetric matrix $M\in\mathbb R^{n\times n}$, we order its
eigenvalues as $\lambda_1(M)\leq\cdots\leq\lambda_n(M)$
 and denote the
corresponding orthonormal eigenvectors by
$z_1(M),\ldots,z_n(M)$.

We use the standard Landau notation for scalar- or vector-valued functions.
Suppose that $g(t)\neq0$ sufficiently close to a specified limit point. As
$t$ approaches that point, we write $f(t)=O(g(t))$ if there exist a constant
$C>0$ and a neighborhood of the limit point such that
$\|f(t)\|\leq C|g(t)|$. We write $f(t)=o(g(t))$ if
$\|f(t)\|/|g(t)|\to0$. Consequently, $f(t)=o(g(t))$ implies
$f(t)=O(g(t))$. We state the relevant limit, such as $t\to\infty$ or
$t\to0$, explicitly unless it is clear from the context.


\subsection{Auxiliary lemmas}

\begin{lemm}[\cite{chung1997spectral}]
    Let $M\in\mathbb R^{n\times n}$ be symmetric. Then, for every
    $x\in\mathbb R^n$,
    \begin{align*}
        \lambda_1(M)\|x\|^2\leq x^TMx
        \leq\lambda_n(M)\|x\|^2.
    \end{align*}
    Moreover, for $i=2,\ldots,n$,
    \begin{align*}
        \lambda_i(M)
        =\min_{\substack{x\neq0,\;
        x\perp z_1(M),\ldots,z_{i-1}(M)}}
        \frac{x^TMx}{x^Tx}.
    \end{align*}
    \label{lemm:varEig}
\end{lemm}

\begin{lemm}[\cite{kurtz1992sufficient}]
    Let $P_p(\theta)\triangleq\sum_{l=0}^p a_l\theta^l$, where $p\geq2$
    and $a_l>0$ for $l=0,\ldots,p$. If
    \begin{align}
        a_l^2-4a_{l-1}a_{l+1}>0,\qquad l=1,\ldots,p-1,
        \label{eq:75}
    \end{align}
    then all roots of $P_p$ are real and distinct.
    \label{lemm:kurtz}
\end{lemm}

\section{Problem Formulation}
\label{sec:probForm}

We consider a network of $n\in\mathbb Z^+$ agents governed by the
single-integrator dynamics
\begin{align}
    \dot{x}_i(t)=u_i(t),\qquad i=1,\ldots,n,
    \label{eq:agent_dynamics}
\end{align}
where $x_i(t)\in\mathbb R$ and $u_i(t)\in\mathbb R$ represent the state and
control input of agent $i$, respectively. Equivalently, we write \eqref{eq:agent_dynamics} in vector form as
\begin{align*}
    \dot{x}(t)=u(t),
\end{align*}
where $x(t)\triangleq [x_1(t),\cdots,x_n(t)]^T$ and $ u(t)\triangleq [u_1(t), \cdots, u_n(t)]^T$. The single-integrator model \eqref{eq:agent_dynamics} isolates the role of controller order from
additional agent dynamics, thereby allowing us to focus on how much integral
action is required for agents. Likewise, we
use scalar states in \eqref{eq:agent_dynamics} only for notational simplicity;
the same formulation extends to higher-dimensional states through Kronecker
products.

We model communication among the agents using an undirected graph
$\mathcal G=(\mathcal V,\mathcal E)$, where
$\mathcal V=\{1,\ldots,n\}$ denotes the set of agents and 
$\mathcal E\subseteq
    \big\{\{i,j\}:i,j\in\mathcal V,\ i\neq j\big\}$
denotes the set of communication links. We define the adjacency matrix
$A=[a_{ij}]\in\mathbb R^{n\times n}$ by setting $a_{ij}=1$ if
$\{i,j\}\in\mathcal E$ and $a_{ij}=0$ otherwise. Because $\mathcal G$ is
undirected, $a_{ij}=a_{ji}$. For each agent~$i$, we define its neighbor set as
$\mathcal N_i=\{j\in\mathcal V:a_{ij}=1\}$ and its degree as
$d_i=\sum_{j=1}^n a_{ij}$.
Using these degrees, we construct the degree matrix $D$ and the graph
Laplacian $L$ as
\begin{align*}
    D=\operatorname{diag}(d_1,\ldots,d_n),
    \qquad L=D-A.
\end{align*}

Let $x_0(t)\in\mathbb R$ denote the state of a dynamic target whose trajectory
specifies the common time-varying motion of the desired multi-agent
configuration. To indicate whether agent $i$ has access to the target state,
we define $a_{i0}\in\{0,1\}$ such that $a_{i0}=1$ if agent $i$ can access
$x_0(t)$ and $a_{i0}=0$ otherwise. We collect these access indicators in the
diagonal matrix
\begin{align*}
    A_0=\operatorname{diag}(a_{10},\ldots,a_{n0}).
\end{align*}

We impose the following standard connectivity and target-access assumption.
\begin{assump}
    The graph $\mathcal G$ is connected, and at least one agent has access to
    the target; equivalently, $\sum_{i=1}^n a_{i0}>0$.
    \label{assump_graph}
\end{assump}

Under Assumption~\ref{assump_graph}, $L+A_0$ is symmetric and positive
definite~\cite{olfati2004consensus}. Throughout the remainder of this paper, eigenvalues written without
an explicit matrix argument refer to those of $L+A_0$; namely,
\begin{align*}
    \lambda_i\triangleq\lambda_i(L+A_0),\qquad
    0<\lambda_1\leq\cdots\leq\lambda_n.
\end{align*}

For each communication edge $\{i,j\}\in\mathcal E$, we prescribe
$\bar\zeta_{ij}$ as the desired value of the relative displacement
$x_i(t)-x_j(t)$ and impose $\bar\zeta_{ji}=-\bar\zeta_{ij}$ so that the
requirement remains independent of the ordering of the two agents. For each
agent $i$ with $a_{i0}=1$, we similarly prescribe $\bar\zeta_{i0}$ as its
desired displacement from the target. To represent both types of requirements
in a unified form, we treat the target as node $0$ and define the augmented
graph
\begin{align*}
    \mathcal G^{\mathrm{augm}}
    &=(\mathcal V^{\mathrm{augm}},\mathcal E^{\mathrm{augm}}),\\
    \mathcal V^{\mathrm{augm}}
    &=\{0\}\cup\mathcal V,\\
    \mathcal E^{\mathrm{augm}}
    &=\mathcal E\cup
    \big\{\{i,0\}:i\in\mathcal V,\ a_{i0}=1\big\}.
\end{align*}
Thus, for every edge $\{i,j\}\in\mathcal E^{\mathrm{augm}}$, the constant
$\bar\zeta_{ij}$ specifies the desired value of $x_i(t)-x_j(t)$.

We now formulate the problem of achieving dynamic coordination through
linear distributed control.
\begin{probl}
    \label{probl_1}
    Given the agent dynamics~\eqref{eq:agent_dynamics} and a general smooth
    target $x_0(t)$, design linear distributed control laws
    $u_i(t)$, $i=1,\ldots,n$, using only $x_i(t)-x_j(t)$ and
    $\bar\zeta_{ij}$ for $j\in\mathcal N_i$, together with
    $x_i(t)-x_0(t)$ and $\bar\zeta_{i0}$ when $a_{i0}=1$, such that the
    resulting closed-loop system satisfies
    \begin{align*}
        \lim_{t\to\infty}
        \left|x_i(t)-x_j(t)-\bar\zeta_{ij}\right|=0,
        \qquad
        \{i,j\}\in\mathcal E^{\mathrm{augm}}.
    \end{align*}
\end{probl}



Problem~\ref{probl_1} provides a unified formulation of several classical
multi-agent coordination problems. If $\bar\zeta_{ij}=0$ for every
$\{i,j\}\in\mathcal E^{\mathrm{augm}}$, all agents asymptotically track the
target, and the problem reduces to leader-following consensus
\cite{hong2006tracking,jadbabaie2003coordination}. If the desired displacements 
specify a nontrivial configuration, the problem reduces to formation control
for a stationary target or formation tracking for a moving target
\cite{fax2004information}. In the higher-dimensional extension, if the
target-relative displacements place the agents around the target and the
agent--agent displacements specify the corresponding surrounding geometry,
the problem becomes a target-surrounding control problem
\cite{chen2010surrounding}. The next section characterizes when the desired
displacements define a consistent coordination pattern.

\section{Consistency}

Arbitrary assignments of the desired relative displacements may be mutually
inconsistent, rendering Problem~\ref{probl_1} infeasible. The following
example illustrates this issue.
\begin{examp}
    \label{examp:1}
    Consider the augmented network of four agents shown in Fig.~\ref{Fig1}.
    \begin{figure}[ht]
  \centering
\begin{tikzpicture}[-,>=stealth',shorten >=1pt,auto,node distance=1.8cm,
  thick,main node/.style={circle,scale=1.2,draw,font=\sffamily\tiny\bfseries}, root node/.style={circle,dashed,draw,font=\sffamily\tiny\bfseries}]
  \tikzset{edge/.style = {->,> = latex'}}
  \node[main node] (1) {1};
  \node[main node] (2) [left of=1] {2};
  \node[main node] (3) [below of=2] {3};
  \node[main node] (4) [right of=3] {4};
  \node[root node] (5) [right of=1] {$x_0(t)$};


  \path[every node/.style={font=\sffamily\small}]
    (1) edge[->] node {$\bar\zeta_{12}$} (2)
    (1) edge[->] node {$\bar\zeta_{14}$} (4)
    (3) edge[->] node {$\bar\zeta_{32}$} (2)
    (3) edge[->] node {$\bar\zeta_{34}$} (4)
    (3) edge[->] node {$\bar\zeta_{31}$} (1)
    (1) edge[->,dashed] node {$\bar\zeta_{10}$} (5);
\end{tikzpicture}
  \caption{\label{Fig1} An inconsistent set of desired relative displacements.}
\end{figure}
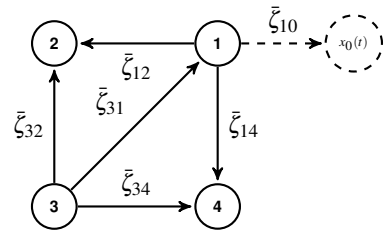
    The solid nodes represent the agents, and the dashed node represents the
    dynamic target. The arrows specify the orientation used to define each
    desired relative displacement. Let $\bar\zeta_{12}=\bar\zeta_{32}=1$ and
    $\bar\zeta_{31}=2$. The first two equalities require
    $x_1-x_2=x_3-x_2$ and hence $x_3-x_1=0$, which contradicts
    $\bar\zeta_{31}=2$. Therefore, these desired displacements are
    inconsistent, and Problem~\ref{probl_1} has no solution.
\end{examp}

Example~\ref{examp:1} raises a natural question: under what conditions do the
prescribed relative displacements make Problem~\ref{probl_1} solvable? The
following lemma relates this edge-wise coordination objective to the tracking
of consistent desired trajectories.
\begin{lemm}
    \label{prop_nec_suff}
    Suppose Assumption~\ref{assump_graph} holds and there exist constants
    $\bar d_0,\ldots,\bar d_n$, with $\bar d_0=0$, such that
    \begin{align}
        \label{eq_57}
        \bar{d}_i-\bar{d}_j=\bar\zeta_{ij},
        \qquad \{i,j\}\in\mathcal E^{\mathrm{augm}}.
    \end{align}
    Define
    \begin{align}
        \bar{x}_i(t) \triangleq x_0(t)+\bar{d}_i,
        \qquad i=0,1,\dots,n.
        \label{eq_20}
    \end{align}
    Problem~\ref{probl_1} is solved if and only if 
    \begin{align*}
        \lim_{t \to \infty}|x_i(t)-\bar{x}_i(t)|=0,
        \qquad i=1,\dots,n.
    \end{align*}
\end{lemm}
\begin{IEEEproof}
    \emph{(Necessity)}
    Starting from the agents that have access to the target, define the layers
    \begin{align}
        \mathcal{H}_1
        &\triangleq \{i\in\mathcal{V}:a_{i0}>0\}, \nonumber\\
        \mathcal{H}_p
        &\triangleq
        \left\{i\in\mathcal{V}\setminus
        \left(\bigcup_{q=1}^{p-1}\mathcal{H}_q\right):
        \exists j\in\mathcal{H}_{p-1}\text{ such that }a_{ij}>0\right\},
        \label{eq_21}
    \end{align}
    for $p=2,\dots,h$, where $h$ denotes the number of nonempty layers.
    Assumption~\ref{assump_graph} ensures that these layers partition the agent
    set, i.e., $\bigcup_{p=1}^h\mathcal H_p=\mathcal V$ and
    $\mathcal H_p\cap\mathcal H_q=\emptyset$ for $p\neq q$.

    We now proceed by induction over the layers. For any
    $i\in\mathcal H_1$, the coordination objective on the edge $\{i,0\}$
    yields
    \begin{align}
        x_i(t)-x_0(t)\to\bar\zeta_{i0}.
        \label{eq_51}
    \end{align}
    By~\eqref{eq_57}, $\bar d_i=\bar\zeta_{i0}$ because $\bar d_0=0$.
    Therefore,
    \begin{align}
        x_i(t)\to x_0(t)+\bar\zeta_{i0}
        =x_0(t)+\bar{d}_i=\bar{x}_i(t),
        \label{eq_50}
    \end{align}
    which establishes the result for the first layer.

    Suppose the result holds for every agent in $\mathcal H_p$, namely,
    \begin{align}
        x_i(t)\to\bar{x}_i(t),\qquad i\in\mathcal H_p.
        \label{eq_52}
    \end{align}
    For any $i\in\mathcal H_{p+1}$, the definition of the layers guarantees
    the existence of a neighbor $j\in\mathcal H_p\cap\mathcal N_i$. The
    coordination objective on the edge $\{i,j\}$ then gives
    \begin{align}
        x_i(t)-x_j(t)-\bar\zeta_{ij}\to0.
        \label{eq_54}
    \end{align}
    By the induction hypothesis,
    \begin{align}
        x_j(t)\to\bar{x}_j(t)=x_0(t)+\bar d_j.
        \label{eq_53}
    \end{align}
    Combining~\eqref{eq_53} and~\eqref{eq_54} yields
    \begin{align*}
        x_i(t)\to x_0(t)+\bar d_j+\bar\zeta_{ij}.
    \end{align*}
    Since~\eqref{eq_57} implies
    $\bar d_i=\bar d_j+\bar\zeta_{ij}$, we obtain
    $x_i(t)\to\bar x_i(t)$. This completes the induction and proves necessity.

    \emph{(Sufficiency)}
    Suppose $x_i(t)-\bar x_i(t)\to0$ for all $i=1,\dots,n$. Moreover,
    $x_0(t)-\bar x_0(t)\equiv0$ by~\eqref{eq_20}. Hence, for every
    $\{i,j\}\in\mathcal E^{\mathrm{augm}}$,
    \begin{align*}
        x_i(t)-x_j(t)-\bar\zeta_{ij}
        &=[x_i(t)-\bar x_i(t)]-[x_j(t)-\bar x_j(t)]\\
        &\quad+\bar d_i-\bar d_j-\bar\zeta_{ij}\to0,
    \end{align*}
    where the last equality follows from~\eqref{eq_57}. Therefore, the
    coordination objective in Problem~\ref{probl_1} holds on every augmented
    edge, proving sufficiency.
\end{IEEEproof}

Lemma~\ref{prop_nec_suff} shows that solving Problem~\ref{probl_1} requires
each agent trajectory $x_i(t)$ to converge to the desired trajectory
$\bar x_i(t)=x_0(t)+\bar d_i$. The constants $\bar d_i$ must reproduce all
prescribed relative displacements through~\eqref{eq_57}. This requirement
motivates the following definition.
\begin{defin}[Consistency]
    The desired relative displacements in Problem~\ref{probl_1} are
    \emph{consistent} if there exist constants $\bar d_0,\ldots,\bar d_n$,
    with $\bar d_0=0$, that satisfy~\eqref{eq_57}.
\end{defin}

We next characterize consistency using the incidence matrix. Arbitrarily orient
and index the edges of $\mathcal G^{\mathrm{augm}}$, and let
$E^{\mathrm{augm}}\in\mathbb R^{(n+1)\times|\mathcal E^{\mathrm{augm}}|}$ be
the corresponding incidence matrix. For an edge oriented from node $i$ to
node $j$, its column has $1$ in row $i$, $-1$ in row $j$, and zeros
elsewhere. Let
$E_x^{\mathrm{augm}}\in\mathbb R^{n\times|\mathcal E^{\mathrm{augm}}|}$ be
obtained by deleting the target-node row, and let $\bar\zeta$ collect the
desired displacements in the same edge order, with component
$\bar\zeta_{ij}$ for an edge oriented from $i$ to $j$.

For the augmented graph in Fig.~\ref{Fig1}, use the node order
$(0,1,2,3,4)$ and the oriented-edge order
$(1,2),(1,4),(3,2),(3,4),(3,1),(1,0)$. Then,
\begin{align*}
E^{\mathrm{augm}}
=
\begin{bmatrix}
0&0&0&0&0&-1\\
1&1&0&0&-1&1\\
-1&0&-1&0&0&0\\
0&0&1&1&1&0\\
0&-1&0&-1&0&0
\end{bmatrix},
\qquad
\bar\zeta=
\begin{bmatrix}
\bar\zeta_{12}\\
\bar\zeta_{14}\\
\bar\zeta_{32}\\
\bar\zeta_{34}\\
\bar\zeta_{31}\\
\bar\zeta_{10}
\end{bmatrix}.
\end{align*}
Deleting the first row gives
\begin{align*}
E_x^{\mathrm{augm}}
=
\begin{bmatrix}
1&1&0&0&-1&1\\
-1&0&-1&0&0&0\\
0&0&1&1&1&0\\
0&-1&0&-1&0&0
\end{bmatrix}.
\end{align*}

The following theorem provides an algebraic characterization of consistency.
\begin{theor}
    \label{theor:consistency}
    The desired relative displacements are consistent if and only if
    \begin{align*}
        \bar\zeta\in\operatorname{range}
        \big((E_x^{\mathrm{augm}})^T\big).
    \end{align*}
\end{theor}
\begin{IEEEproof}
    Let $\bar d=[\bar d_1,\dots,\bar d_n]^T$. Stacking~\eqref{eq_57} in the
    chosen edge order and using $\bar d_0=0$ yields
    \begin{align}
        (E_x^{\mathrm{augm}})^T\bar d=\bar\zeta.
        \label{eq_56}
    \end{align}
    This equation is solvable if and only if
    $\bar\zeta\in\operatorname{range}((E_x^{\mathrm{augm}})^T)$.
\end{IEEEproof}

    The condition in Theorem~\ref{theor:consistency} has a direct
    graph-theoretic interpretation: the signed sum of the desired
    displacements along every cycle of $\mathcal G^{\mathrm{augm}}$ must be
    zero. For example, the cycle formed by agents $1$, $2$, and $3$ in
    Fig.~\ref{Fig1} requires
    $\bar\zeta_{12}-\bar\zeta_{32}+\bar\zeta_{31}=0$, which is violated by
    the values in Example~\ref{examp:1}. Moreover, under
    Assumption~\ref{assump_graph}, $E_x^{\mathrm{augm}}$ has full row rank;
    hence, whenever the displacements are consistent, the offset vector
    $\bar d$ satisfying~\eqref{eq_56} is unique.

\section{Algorithm}
\label{sec:adaptive_order}

Building on the coordination objective established above, we now
develop OADIC and its distributed performance monitor.
OADIC starts at the lowest controller order and activates an additional
integral state whenever the regulation error fails to decrease by a prescribed
percentage over a decision interval. Thus, order adequacy is assessed from the
relative improvement achieved by the current controller, rather than from the
absolute magnitude of the error. This process continues until the controller
reaches a specified maximum order.

Suppose that the desired displacements in Problem~\ref{probl_1} are
consistent. We define the desired network trajectory and the corresponding
tracking error as
\begin{align}
    \bar{x}(t)\triangleq\mathbf{1}x_0(t)+\bar{d},\qquad
    e(t)\triangleq x(t)-\bar{x}(t).
    \label{eq:smooth_reference_error}
\end{align}
Here $\bar{d}=[\bar{d}_1,\dots,\bar{d}_n]^T$ denotes the offset vector that
satisfies~\eqref{eq_57}. The stacked proportional error is then given by
\begin{align}
    \sigma_0(t)\triangleq (L+A_0)e(t).
    \label{eq:smooth_sigma0}
\end{align}

Define
$\mathcal{Q}\triangleq\{1,\dots,q_{\max}\}$, where the prescribed maximum
order satisfies $q_{\max}\geq2$. This lower bound ensures that at least one
order increase is possible: $q=1$ corresponds to proportional feedback,
whereas $q=2$ activates the first integral state. For each $q\in\mathcal Q$, the order-$q$
controller contains $q-1$ integral states, and agent $i$ implements
\begin{align}
    u_i(t)&=-\sum_{l=0}^{q-1}\kappa_l\sigma_{il}(t),
    \label{eq:smooth_control}\\
    \dot{\sigma}_{il}(t)&=\sigma_{i,l-1}(t),
    \qquad l=1,\dots,q-1,
    \label{eq:smooth_integrators}
\end{align}
where $\kappa_l>0$ and
\begin{align}
    \sigma_{i0}(t)
    &=\sum_{j\in\mathcal{N}_i}
    \big[x_i(t)-x_j(t)-\bar\zeta_{ij}\big]+a_{i0}\big[x_i(t)-x_0(t)-\bar\zeta_{i0}\big].
    \label{eq:smooth_agent_proportional_error}
\end{align}
Agent $i$ can compute every term
in~\eqref{eq:smooth_agent_proportional_error} locally using relative
measurements from its neighbors and, when $a_{i0}=1$, its relative measurement
with respect to the target.

To assess the regulation performance locally, agent $i$ computes
\begin{align}
    \rho_i^2(t)
    \triangleq\frac{1}{2}\sum_{j\in\mathcal{N}_i}
    \big[x_i(t)-x_j(t)-\bar\zeta_{ij}\big]^2+a_{i0}\big[x_i(t)-x_0(t)-\bar\zeta_{i0}\big]^2.
    \label{eq:smooth_local_error}
\end{align}
The factor $1/2$ compensates for double counting of the agent--agent terms.
Indeed, summing~\eqref{eq:smooth_local_error} over all agents counts the
squared error for each undirected edge $\{i,j\}$ once through $\rho_i^2(t)$
and once through $\rho_j^2(t)$. The factor $1/2$ therefore gives this error a
total coefficient of one. By contrast, only the agent with access to the
corresponding target measurement contributes each target-related term.

We define the network-wide regulation indicator as
\begin{align}
    J(t)\triangleq\max_{i=1,\dots,n}\rho_i(t).
    \label{eq:smooth_switching_indicator}
\end{align}
Let $0=t_0<t_1<t_2<\cdots$ denote the decision instants, with $t_k\to\infty$
and $t_{k+1}-t_k\geq T_d>0$, where $T_d$ specifies the minimum interval
between consecutive order decisions. To quantify the improvement achieved
during $[t_k,t_{k+1}]$, the agents evaluate the network-wide indicator at
the beginning and end of the interval. Define the relative error reduction
over this interval as
\begin{align}
    r_k\triangleq
    \begin{cases}
        \displaystyle
        \frac{J(t_k^+)-J(t_{k+1}^-)}{J(t_k^+)},
        &J(t_k^+)>0,\\[3mm]
        1,&J(t_k^+)=J(t_{k+1}^-)=0,\\
        -\infty,&J(t_k^+)=0<J(t_{k+1}^-).
    \end{cases}
    \label{eq:smooth_relative_reduction}
\end{align}
Thus, $100r_k\%$ is the percentage by which the monitored error decreases
during the interval. In particular, $r_k<0$ indicates error growth, while
the convention $r_k=1$ at zero error prevents an unnecessary order increase
when exact regulation has already been achieved.

\begin{rem}
    \label{rem:smooth_distributed_max}
    The endpoint values $J(t_k^+)$ and $J(t_{k+1}^-)$ can be computed in a
    fully distributed manner. At each endpoint, agent $i$ initializes a
    finite-time max-consensus algorithm with its local value $\rho_i(t_k^+)$
    or $\rho_i(t_{k+1}^-)$ and exchanges only its current estimate with
    neighboring agents. Because $\mathcal G$ is connected by
    Assumption~\ref{assump_graph}, all agents obtain the same network-wide
    maximum after the max-consensus procedure
    terminates~\cite{furchi2025finite}. 
\end{rem}

The following lemma establishes that the network-wide indicator $J(t)$,
although constructed from the local quantities $\rho_i(t)$, provides an
equivalent measure of the global tracking error.
\begin{lemm}
    \label{lemm:smooth_monitor_equivalence}
    Under Assumption~\ref{assump_graph},
    \begin{align*}
        \sum_{i=1}^n\rho_i^2(t)&=e^T(t)(L+A_0)e(t),\\
        \sqrt{\frac{\lambda_1}{n}}\|e(t)\|
        &\leq J(t)\leq
        \sqrt{\lambda_n}\|e(t)\|.
    \end{align*}
\end{lemm}

\begin{IEEEproof}
    By~\eqref{eq:smooth_reference_error}, $ x_i(t)=e_i(t)+x_0(t)+\bar d_i.$
    Moreover, the consistency condition together with~\eqref{eq_57} implies
    $\bar d_i-\bar d_j=\bar\zeta_{ij}$ for every agent edge and
    $\bar d_i=\bar\zeta_{i0}$ whenever $a_{i0}>0$. Therefore,
    \begin{align*}
        x_i(t)-x_j(t)-\bar\zeta_{ij}
        &=e_i(t)-e_j(t)+\bar d_i-\bar d_j-\bar\zeta_{ij}\\
        &=e_i(t)-e_j(t),\\
        x_i(t)-x_0(t)-\bar\zeta_{i0}
        &=e_i(t)+\bar d_i-\bar\zeta_{i0}=e_i(t).
    \end{align*}
    Substituting these
    identities into~\eqref{eq:smooth_local_error} and summing over all agents
    yields
    \begin{align*}
        \sum_{i=1}^n\rho_i^2(t)
        &=\frac{1}{2}\sum_{i=1}^n\sum_{j\in\mathcal N_i}
        [e_i(t)-e_j(t)]^2+e^T(t)A_0e(t)\\
        &=e^T(t)(L+A_0)e(t).
    \end{align*}
    This proves the first equality in
    Lemma~\ref{lemm:smooth_monitor_equivalence}.
    
    Since $J(t)=\max_i\rho_i(t)$, we also have
    \begin{align*}
        J^2(t)\leq\sum_{i=1}^n\rho_i^2(t)\leq nJ^2(t).
    \end{align*}
    Finally, applying Lemma~\ref{lemm:varEig} to $L+A_0$ gives
    \begin{align*}
        \lambda_1\|e(t)\|^2
        \leq e^T(t)(L+A_0)e(t)
        \leq \lambda_n\|e(t)\|^2.
    \end{align*}
    Using the established identity in the preceding two inequalities yields
    $\lambda_1\|e(t)\|^2\leq nJ^2(t)$ and
    $J^2(t)\leq\lambda_n\|e(t)\|^2$. Taking square roots gives
    $\sqrt{\lambda_1/n}\|e(t)\|\leq J(t)\leq
    \sqrt{\lambda_n}\|e(t)\|$, as stated in
    Lemma~\ref{lemm:smooth_monitor_equivalence}.
\end{IEEEproof}

\begin{rem}
    Lemma~\ref{lemm:smooth_monitor_equivalence} shows that $J(t)$ is a
    distributedly computable measure equivalent to the global tracking error
    $\|e(t)\|$. In particular, $J(t)=0$ if and only if $e(t)=0$. Therefore,
    the controller order can be selected using the local regulation quantities
    $\rho_i(t)$ and a network-wide maximum operation without explicitly
    reconstructing the global error vector $e(t)$. 
    \label{rem:smooth_monitor_equivalence}
\end{rem}

Let $q_k\triangleq q(t_k^+)$ denote the controller order applied over
$[t_k,t_{k+1})$. To determine whether this order remains adequate, choose
$\varepsilon_0>0$, $\varepsilon_\infty\geq0$, and $\alpha>0$ such that
$\varepsilon_0+\varepsilon_\infty<1$, and define the required relative
reduction sequence
\begin{align}
    \varepsilon_k
    \triangleq\varepsilon_\infty+
    \frac{\varepsilon_0}{(k+1)^\alpha}.
    \label{eq:smooth_threshold_schedule}
\end{align}
Here, $100\varepsilon_k\%$ is the minimum percentage reduction expected over
the $k$th decision interval. The parameter $\varepsilon_\infty$ specifies the
long-run minimum improvement, $\varepsilon_0$ adds a stricter initial
requirement, and $\alpha$ controls how quickly this additional requirement
decays. 

Starting from $q_0=1$, update the common controller order at $t_{k+1}$
according to
\begin{align}
    q_{k+1}&=
    \begin{cases}
        q_k+1,&\text{if }r_k\leq\varepsilon_k
        \text{ and }q_k<q_{\max},\\
        q_k,&\text{otherwise},
    \end{cases}
    \label{eq:smooth_switching_rule}
\end{align}
Under~\eqref{eq:smooth_switching_rule}, a relative reduction no greater than
$\varepsilon_k$ indicates that the current internal-model order has not
improved the tracking error sufficiently during the decision interval. OADIC
therefore activates one additional integral state, provided that
$q_k<q_{\max}$. If the achieved reduction is larger than the required
percentage, or if the maximum order has already been reached, the order
remains unchanged. 

We specify the controller-state transfer accompanying an order increase as
follows. At a switching instant $t_s$ from order $q$ to order $q+1$, retain
all existing integral states and initialize only the newly activated state:
\begin{align}
    \sigma_l(t_s^+)&=\sigma_l(t_s^-),\quad l=1,\dots,q-1,\nonumber\\
    \sigma_q(t_s^+)&=0,\qquad\dot{\sigma}_q(t)=\sigma_{q-1}(t).
    \label{eq:smooth_reset}
\end{align}
Because both $x(t)$ and $x_0(t)$ are continuous at $t_s$, the tracking error
also remains continuous, i.e., $e(t_s^+)=e(t_s^-)$. It then follows from
\eqref{eq:smooth_sigma0} that
$\sigma_0(t_s^+)=\sigma_0(t_s^-)$. Combining this identity with the reset
rule~\eqref{eq:smooth_reset}, the order-$(q+1)$ control input immediately
after the switch is
\begin{align}
    u(t_s^+)
    &=-\sum_{l=0}^{q}\kappa_l\sigma_l(t_s^+)\nonumber\\
    &=-\sum_{l=0}^{q-1}\kappa_l\sigma_l(t_s^-)
    -\kappa_q\sigma_q(t_s^+)\nonumber\\
    &=-\sum_{l=0}^{q-1}\kappa_l\sigma_l(t_s^-)=u(t_s^-).
    \label{eq:smooth_bumpless}
\end{align}
Thus, $u(t_s^+)=u(t_s^-)$, and the proposed state-transfer rule yields a
bumpless change in controller order.

\begin{rem}
    \label{rem:smooth_distributed}
    OADIC, consisting of the controller
    \eqref{eq:smooth_control}--\eqref{eq:smooth_integrators}, the local monitor
    \eqref{eq:smooth_local_error}--\eqref{eq:smooth_relative_reduction}, the
    switching law~\eqref{eq:smooth_switching_rule}, and the reset
    \eqref{eq:smooth_reset}, remains distributed. Agent $i$ computes
    $\rho_i(t)$ from local relative measurements and records its endpoint
    values over $[t_k,t_{k+1}]$, while $J(t_k^+)$ and $J(t_{k+1}^-)$ are
    obtained distributively as described in
    Remark~\ref{rem:smooth_distributed_max}. Each agent then evaluates the
    same relative reduction $r_k$ and applies the same order update.
\end{rem}

It remains to choose the controller gains $\kappa_l$. Because the controller order is
adjusted online, designing an unrelated set of gains for each admissible order
would require gain replacement at every switch and would complicate the
stability analysis. We therefore use a single nested gain sequence, so that an
order increase only activates the next gain while leaving all previously used
gains unchanged. The sequence is selected constructively to provide a common
sufficient stability condition for every fixed-order subsystem. Specifically,
choose
\begin{align}
    \kappa_0&>0,\qquad
    0<\kappa_1<\frac{\lambda_1}{4}\kappa_0^2,\nonumber\\
    0<\kappa_{l+1}&<\frac{\kappa_l^2}{4\kappa_{l-1}},
    \quad l=1,\dots,q_{\max}-2,\quad\text{if }q_{\max}\geq3.
    \label{eq:smooth_gain_design}
\end{align}
Thus, the order-$q$ controller uses the prefix
$\kappa_0,\dots,\kappa_{q-1}$ of the same sequence. When $q_{\max}=2$, only
the first line of~\eqref{eq:smooth_gain_design} is needed. When
$q_{\max}\geq3$, the last recursion, obtained by setting
$l=q_{\max}-2$, determines the admissible range of
$\kappa_{q_{\max}-1}$. 
Distributed Laplacian-spectrum estimation
algorithms allow agents to estimate $\lambda_1$ through neighbor-to-neighbor
exchanges~\cite{franceschelli2013decentralized}. 
Alternatively, any certified
lower bound $\underline{\lambda}_1\leq\lambda_1$ can replace $\lambda_1$ in
\eqref{eq:smooth_gain_design}, yielding a more conservative but still
sufficient gain selection.

For the analysis, let
\begin{align*}
    \eta(t)&\triangleq
    \begin{bmatrix}
    e^T(t)&\cdots&[e^{(q(t)-1)}(t)]^T
    \end{bmatrix}^T,\\
    w(t)&\triangleq
    \begin{bmatrix}\dot x_0(t)&\cdots&x_0^{(q_{\max})}(t)
    \end{bmatrix}^T,
\end{align*}
and define
\begin{align*}
    \mathcal A_1&\triangleq-\kappa_0(L+A_0),\qquad
    \mathcal B_1\triangleq-\mathbf{1},\\
    \mathcal A_q&\triangleq
    \begin{bmatrix}
    0_{n(q-1)\times n}&I_{n(q-1)}\\
    -\kappa_{q-1}(L+A_0)&
    -\begin{bmatrix}\kappa_{q-2}&\cdots&\kappa_0\end{bmatrix}
    \otimes(L+A_0)
    \end{bmatrix},\\
    &\hspace{32mm}q=2,\dots,q_{\max},\\
    \mathcal B_q&\triangleq
    \begin{bmatrix}0&0&\cdots&0&-\mathbf{1}^T\end{bmatrix}^T,
    \quad q=2,\dots,q_{\max}.
\end{align*}
Between two consecutive switching instants, $q(t)=q$ is constant. Stacking
\eqref{eq:smooth_control}, and using
\eqref{eq:smooth_reference_error} and
$\dot{\bar x}(t)=\mathbf1\dot x_0(t)$, gives
\begin{align}
    \dot e(t)=-\sum_{l=0}^{q-1}\kappa_l\sigma_l(t)
    -\mathbf1\dot x_0(t).
    \label{eq:smooth_error_first_derivative}
\end{align}
Differentiating~\eqref{eq:smooth_error_first_derivative} $q-1$ times gives
\begin{align}
    e^{(q)}(t)=-\sum_{l=0}^{q-1}\kappa_l
    \sigma_l^{(q-1)}(t)-\mathbf1x_0^{(q)}(t).
    \label{eq:smooth_error_q_derivative}
\end{align}
Equation~\eqref{eq:smooth_sigma0} and repeated use of
\eqref{eq:smooth_integrators} give
\begin{align}
    \begin{aligned}
        \sigma_0^{(q-1)}(t)&=(L+A_0)e^{(q-1)}(t),\\
        \sigma_l^{(q-1)}(t)&=\sigma_0^{(q-1-l)}(t)
        =(L+A_0)e^{(q-1-l)}(t),\\[-1mm]
        &\hspace{39mm}l=1,\dots,q-1.
    \end{aligned}
    \label{eq:smooth_sigma_derivatives}
\end{align}
Substituting~\eqref{eq:smooth_sigma_derivatives} into
\eqref{eq:smooth_error_q_derivative} gives
\begin{align}
    e^{(q)}(t)+\sum_{l=0}^{q-1}\kappa_l(L+A_0)e^{(q-1-l)}(t)
    =-\mathbf 1x_0^{(q)}(t).
    \label{eq:smooth_error_dynamics}
\end{align}
By the definition of $\eta$,
\begin{align*}
    \dot\eta(t)
    =\begin{bmatrix}
    [e^{(1)}(t)]^T&\cdots&[e^{(q-1)}(t)]^T&[e^{(q)}(t)]^T
    \end{bmatrix}^T.
\end{align*}
The first $q-1$ block rows of $\mathcal A_q\eta(t)$ reproduce
$e^{(1)}(t),\dots,e^{(q-1)}(t)$, while its last block row is the left-hand side
of~\eqref{eq:smooth_error_dynamics} solved for $e^{(q)}(t)$. The input term in
that last block row is
$-\mathbf1x_0^{(q)}(t)=\mathcal B_qx_0^{(q)}(t)$. Hence
\begin{align}
    \dot\eta(t)=\mathcal A_q\eta(t)+\mathcal B_qx_0^{(q)}(t).
    \label{eq:smooth_state_equation}
\end{align}

The next lemma establishes a linear reset map that uniquely determines the
post-switch derivative-state vector from its pre-switch value and the target
derivatives when the controller order increases from $q$ to $q+1$.
\begin{lemm}
    \label{lemm:smooth_reset_map}
    For each $q=1,\dots,q_{\max}-1$, there exist matrices
    $\mathcal R_q\in\mathbb R^{n(q+1)\times nq}$ and
    $\mathcal S_q\in\mathbb R^{n(q+1)\times q_{\max}}$ such that, whenever
    the controller order increases from $q$ to $q+1$ at $t_s$,
    \begin{align*}
        \eta(t_s^+)
        =\mathcal R_q\eta(t_s^-)+\mathcal S_qw(t_s).
    \end{align*}
\end{lemm}

\begin{IEEEproof}
    The proof proceeds in three steps. First, we recover the pre-switch
    integral states from $\eta(t_s^-)$ and $w(t_s)$. From
    Equations~\eqref{eq:smooth_sigma0}--\eqref{eq:smooth_integrators} and
    $\dot{\bar x}(t)=\mathbf 1\dot x_0(t)$, we have
    \begin{align*}
        \dot e(t)
        =-\kappa_0(L+A_0)e(t)
        -\sum_{l=1}^{q-1}\kappa_l\sigma_l(t)-\mathbf 1\dot x_0(t).
    \end{align*}
    Differentiating this equation $j-1$ times and using
    $\dot\sigma_l(t)=\sigma_{l-1}(t)$ and
    $\sigma_0(t)=(L+A_0)e(t)$ yields, for
    $j=1,\dots,q-1$,
    \begin{align}
        e^{(j)}(t_s^-)
        &=-\sum_{l=0}^{j-1}\kappa_l(L+A_0)
        e^{(j-1-l)}(t_s^-)\nonumber\\
        &\quad-\sum_{l=j}^{q-1}\kappa_l
        \sigma_{l-j+1}(t_s^-)-\mathbf 1x_0^{(j)}(t_s).
        \label{eq:smooth_pre_switch_derivatives}
    \end{align}
    The quantities $e^{(j)}(t_s^-)$ and $x_0^{(j)}(t_s)$ in this identity are
    components of $\eta(t_s^-)$ and $w(t_s)$, respectively. For $q\geq2$,
    the integral states can therefore be recovered recursively. In
    particular, setting $j=q-1$ yields
    \begin{align*}
        \kappa_{q-1}\sigma_1(t_s^-)
        &=-e^{(q-1)}(t_s^-)
        -\sum_{l=0}^{q-2}\kappa_l(L+A_0)
        e^{(q-2-l)}(t_s^-)\\
        &\quad-\mathbf 1x_0^{(q-1)}(t_s),
    \end{align*}
    which uniquely determines $\sigma_1(t_s^-)$ because
    $\kappa_{q-1}>0$. Now suppose that
    $\sigma_1(t_s^-),\dots,\sigma_{m-1}(t_s^-)$ have already been recovered.
    Setting $j=q-m$ in~\eqref{eq:smooth_pre_switch_derivatives}, the second
    sum becomes
    \begin{align*}
        \kappa_{q-1}\sigma_m(t_s^-)
        +\sum_{l=q-m}^{q-2}\kappa_l
        \sigma_{l-q+m+1}(t_s^-).
    \end{align*}
    All terms in this expression except $\sigma_m(t_s^-)$ are known from the
    preceding recursion. Thus, $\sigma_m(t_s^-)$ is also uniquely determined.
    Repeating this argument for $m=2,\dots,q-1$ recovers all pre-switch
    integral states. If $q=1$, the entire reconstruction step is omitted
    because the controller has no integral state before the switch.

    Second, we construct the post-switch derivative state. The continuity of
    $e$ and the reset rule~\eqref{eq:smooth_reset} imply
    \begin{align*}
        e(t_s^+)&=e(t_s^-),\\
        \sigma_l(t_s^+)&=\sigma_l(t_s^-),\quad l=1,\dots,q-1,\\
        \sigma_q(t_s^+)&=0.
    \end{align*}
    Replacing $q$ with $q+1$ in
    \eqref{eq:smooth_pre_switch_derivatives} and evaluating the resulting
    identity at $t_s^+$ yields, for $j=1,\dots,q$,
    \begin{align}
        e^{(j)}(t_s^+)
        &=-\sum_{l=0}^{j-1}\kappa_l(L+A_0)
        e^{(j-1-l)}(t_s^+) \nonumber \\
        &\quad-\sum_{l=j}^{q}\kappa_l
        \sigma_{l-j+1}(t_s^+)-\mathbf 1x_0^{(j)}(t_s).
        \label{eq:smooth_post_switch_derivatives}
    \end{align}
    Setting $j=1$ in~\eqref{eq:smooth_post_switch_derivatives} gives
    \begin{align}
        e^{(1)}(t_s^+)
        &=-\kappa_0(L+A_0)e(t_s^+)
        -\sum_{l=1}^{q}\kappa_l\sigma_l(t_s^+)
        -\mathbf 1\dot x_0(t_s) \nonumber \\
        &=-\kappa_0(L+A_0)e(t_s^-)
        -\sum_{l=1}^{q-1}\kappa_l\sigma_l(t_s^-)
        -\mathbf 1\dot x_0(t_s),
        \label{eq:smooth_post_switch_first_derivative}
    \end{align}
    where the second equality follows from the continuity of $e$, the
    preservation of the existing integral states, and
    $\sigma_q(t_s^+)=0$. All quantities in the final expression of
    \eqref{eq:smooth_post_switch_first_derivative} are known
    from the first step and $w(t_s)$, so $e^{(1)}(t_s^+)$ is uniquely
    determined. Moreover, comparison with
    \eqref{eq:smooth_error_first_derivative} evaluated at $t_s^-$ shows
    that $e^{(1)}(t_s^+)=e^{(1)}(t_s^-)$. Substituting this result into the
    $j=2$ case of~\eqref{eq:smooth_post_switch_derivatives} determines
    $e^{(2)}(t_s^+)$. Proceeding successively through $j=q$ determines the
    entire vector $\eta(t_s^+)$.

    Finally, both reconstruction steps involve only linear operations.
    Specifically, solving
    \eqref{eq:smooth_pre_switch_derivatives} successively with
    $j=q-1,q-2,\dots,1$ expresses each pre-switch integral state
    $\sigma_m(t_s^-)$, $m=1,\dots,q-1$, as a linear combination of
    $\eta(t_s^-)$ and $w(t_s)$. Next, substituting these expressions and the
    reset relations in~\eqref{eq:smooth_reset} into
    \eqref{eq:smooth_post_switch_derivatives}, and evaluating the latter
    successively for $j=1,\dots,q$, expresses every component of
    $\eta(t_s^+)$ as a linear combination of $\eta(t_s^-)$ and $w(t_s)$.
    Collecting the corresponding coefficients therefore gives matrices
    $\mathcal R_q$ and $\mathcal S_q$ satisfying
    \begin{align*}
        \eta(t_s^+)
        =\mathcal R_q\eta(t_s^-)+\mathcal S_qw(t_s)
    \end{align*}
    at every switch from order $q$ to order $q+1$, which proves the lemma.
\end{IEEEproof}

    Lemma~\ref{lemm:smooth_reset_map} shows that increasing the controller
    order does not introduce an independent jump variable into the
    derivative-state dynamics in \eqref{eq:smooth_state_equation}. Once the pre-switch derivative state and the
    target derivatives at the switching instant are known, the reset
    rule~\eqref{eq:smooth_reset} determines the complete post-switch
    derivative state, which admits a linear representation in terms of these
    quantities. Thus, although the dimension of $\eta$ increases from $nq$ to
    $n(q+1)$, no additional independent state is introduced at the switch.
    This representation connects the bumpless controller implementation with
    the switched state-space stability analysis that follows.
    \label{rem:smooth_reset_map}

\section{Stability}
\label{sec:oadic_stability}

This section establishes the stability of OADIC by
analyzing the continuous dynamics associated with each active controller
order together with the state-reset relation and the relative-reduction
switching law.
\begin{theor}
    \label{theor:adaptive_dhic}
    Consider system~\eqref{eq:agent_dynamics} under OADIC, defined by the
    controller
    \eqref{eq:smooth_control}--\eqref{eq:smooth_integrators}, the monitor
    \eqref{eq:smooth_local_error}--\eqref{eq:smooth_relative_reduction}, the
    switching law~\eqref{eq:smooth_switching_rule}, and the reset
    rule~\eqref{eq:smooth_reset}. Suppose that
    Assumption~\ref{assump_graph} holds, Problem~\ref{probl_1} is consistent,
    $q_{\max}\geq2$, $x_0\in C^{q_{\max}}$, and the gains satisfy
    \eqref{eq:smooth_gain_design}. Then, the following statements hold.
    \begin{enumerate}[(i)]
    \item The state matrix $\mathcal A_q$ in
    \eqref{eq:smooth_state_equation} is Hurwitz for every $q\in\mathcal Q$.

    \item The complete switched derivative-state dynamics, consisting of the
    continuous evolution~\eqref{eq:smooth_state_equation} between switching
    instants and the state jumps induced by the reset rule
    \eqref{eq:smooth_reset}, is uniformly input-to-state stable with state
    $\eta(t)$ and input $w(t)$. In particular,
    \begin{align}
        \|\eta(t)\|&\leq C_{\mathrm{ISS}}\exp\left(-\frac{a}{2}t\right)
        \|\eta(0)\|+\gamma_{\mathrm{ISS}}\sup_{0\leq\tau\leq t}\|w(\tau)\|,
        \label{eq:smooth_switched_bound}
    \end{align}
    where
    \begin{align*}
        C_{\mathrm{ISS}}&\triangleq
        \left(
        \frac{\displaystyle\max_{q\in\mathcal Q}\lambda_{\max}(P_q)}
        {\displaystyle\min_{q\in\mathcal Q}\lambda_{\min}(P_q)}
        \mu^{q_{\max}-1}\right)^{1/2},\\
        \gamma_{\mathrm{ISS}}&\triangleq
        \left(
        \frac{G}
        {\displaystyle\min_{q\in\mathcal Q}\lambda_{\min}(P_q)}
        \right)^{1/2},\\
        P_q&\triangleq\int_0^\infty
        \exp(\mathcal A_q^T\tau)\exp(\mathcal A_q\tau)\mathrm d\tau,\\
        a&\triangleq
        \min_{q\in\mathcal Q}\frac{1}
        {2\lambda_{\max}(P_q)},\\
        \mu&\triangleq
        \max_{1\leq q<q_{\max}}
        \max\left\{1,
        \frac{2\lambda_{\max}
        (\mathcal R_q^TP_{q+1}\mathcal R_q)}
        {\lambda_{\min}(P_q)}\right\},\\
        G&\triangleq
        \frac{2\mu^{q_{\max}-1}}{a}
        \max_{q\in\mathcal Q}\|P_q\mathcal B_q\|^2\nonumber\\
        &\quad+2\max_{1\leq q<q_{\max}}
        \left\|P_{q+1}^{1/2}\mathcal S_q\right\|^2
        \sum_{j=0}^{q_{\max}-2}\mu^j,
    \end{align*}
    and  $\mathcal R_q$ and $\mathcal S_q$ are the reset-map matrices in
    Lemma~\ref{lemm:smooth_reset_map}.

    \item The final order
    $q_\infty\triangleq\lim_{k\to\infty}q_k\in\mathcal Q$ exists. If
    $q_\infty<q_{\max}$, then, for some $K\in\mathbb Z^+$ and all $k\geq K$,
    \begin{align}
        \|e(t_k)\|
        \leq\sqrt{\frac{n}{\lambda_1}}J(t_k)\leq\sqrt{\frac{n}{\lambda_1}}J(t_K)
        \prod_{\ell=K}^{k-1}(1-\varepsilon_\ell).
        \label{eq:smooth_unsaturated_error_bound}
    \end{align}
    Hence, $e(t_k)\to0$ if
    $\sum_{\ell=0}^{\infty}\varepsilon_\ell=\infty$; under
    \eqref{eq:smooth_threshold_schedule}, this holds if
    $\varepsilon_\infty>0$, or if
    $\varepsilon_\infty=0$ and $0<\alpha\leq1$.

    If $q_\infty=q_{\max}$, then
    \begin{align}
        \limsup_{t\to\infty}\|e(t)\|
        &\leq
        \left(\int_0^\infty
        \|\exp(\mathcal A_{q_{\max}}s)\mathcal B_{q_{\max}}\|
        \mathrm ds\right)\nonumber\\
        &\quad\times
        \limsup_{t\to\infty}|x_0^{(q_{\max})}(t)|.
        \label{eq:smooth_saturated_bound}
    \end{align}
    For either final order, $x_0^{(q_\infty)}(t)\to0$ implies $e(t)\to0$.
    \end{enumerate}
\end{theor}

\begin{IEEEproof}
    \emph{(i)}
    By Assumption~\ref{assump_graph}, $L+A_0$ is
    symmetric and positive definite. Hence, there exists an orthogonal matrix
    $U$ such that
    \begin{align*}
        U^T(L+A_0)U=\Lambda\triangleq
        \operatorname{diag}(\lambda_1,\dots,\lambda_n),
    \end{align*}
    where $0<\lambda_1\leq\cdots\leq\lambda_n$ are the eigenvalues of
    $L+A_0$. 
    The Hurwitz property concerns the zero-input dynamics
    $\dot\eta=\mathcal A_q\eta$. Therefore, when analyzing the eigenvalues of
    $\mathcal A_q$, we set the input $x_0^{(q)}(t)$ in
    \eqref{eq:smooth_state_equation} to zero. The corresponding zero-input form of
    \eqref{eq:smooth_error_dynamics} is
    \begin{align}
        e^{(q)}(t)+\sum_{l=0}^{q-1}
        \kappa_l(L+A_0)e^{(q-1-l)}(t)=0.
        \label{eq:smooth_homogeneous_error}
    \end{align}

    Introduce the transformed error $z(t)\triangleq U^Te(t)$. Multiplying
    \eqref{eq:smooth_homogeneous_error} by $U^T$ and using
    $U^T(L+A_0)U=\Lambda$ gives
    \begin{align*}
        z^{(q)}(t)+\sum_{l=0}^{q-1}
        \kappa_l\Lambda z^{(q-1-l)}(t)=0.
    \end{align*}
    Because $\Lambda$ is diagonal, the $n$ coordinates of $z$ are decoupled.
    In particular, the coordinate associated with the eigenvalue $\lambda_i$
    satisfies
    \begin{align}
        z_i^{(q)}(t)
        +\lambda_i\sum_{l=0}^{q-1}
        \kappa_l z_i^{(q-1-l)}(t)=0.
        \label{eq:smooth_scalar_mode}
    \end{align}
    Hence, each transformed coordinate $z_i(t)$ evolves independently of the
    others according to~\eqref{eq:smooth_scalar_mode}.

    To express~\eqref{eq:smooth_scalar_mode} as a first-order system, define
    \begin{align*}
        \xi_i(t)\triangleq
        \begin{bmatrix}
            z_i(t)&\dot z_i(t)&\cdots&z_i^{(q-1)}(t)
        \end{bmatrix}^T.
    \end{align*}
    For $q\geq2$, equation~\eqref{eq:smooth_scalar_mode} is equivalent to
    $\dot\xi_i(t)=C_{q,i}\xi_i(t)$, where
    \begin{align*}
        C_{q,i}\triangleq
        \begin{bmatrix}
            0_{(q-1)\times1}&I_{q-1}\\
            -\kappa_{q-1}\lambda_i&
            -\lambda_i
            \begin{bmatrix}\kappa_{q-2}&\cdots&\kappa_0\end{bmatrix}
        \end{bmatrix}.
    \end{align*}
    When $q=1$, the corresponding state matrix is
    $C_{1,i}=-\kappa_0\lambda_i$. The transformation
    $z(t)=U^Te(t)$, followed by a reordering that groups the derivatives of each
    $z_i(t)$, transforms $\mathcal A_q$ into
    $\operatorname{diag}(C_{q,1},\dots,C_{q,n})$. Thus, the eigenvalues of
    $\mathcal A_q$ are precisely the eigenvalues of these scalar-coordinate
    state matrices.

    To calculate the eigenvalues of $C_{q,i}$, let $s\in\mathbb C$ and seek a
    nonzero solution of~\eqref{eq:smooth_scalar_mode} in the form
    $z_i(t)=c\exp(st)$, where $c\neq0$. Substituting
    $z_i^{(r)}(t)=s^rc\exp(st)$ into
    \eqref{eq:smooth_scalar_mode} and dividing the resulting equation by the
    nonzero factor $c\exp(st)$ gives
    \begin{align}
        p_{q,i}(s)
        &\triangleq s^q+\kappa_0\lambda_i s^{q-1}
        +\kappa_1\lambda_i s^{q-2}
        +\cdots+\kappa_{q-1}\lambda_i=0.
        \label{eq:smooth_modal_polynomial}
    \end{align}
    Therefore, $p_{q,i}(s)=0$ is the characteristic equation of $C_{q,i}$,
    and its roots are the eigenvalues of $C_{q,i}$. 

    It remains to show that all these roots are negative. For $q=1$,
    $p_{1,i}(s)=s+\kappa_0\lambda_i$, whose root is
    $-\kappa_0\lambda_i<0$. Now let $q\geq2$ and write
    \begin{align*}
        p_{q,i}(s)&=\sum_{m=0}^qa_ms^m,\\
        a_q&=1,\qquad
        a_{q-1-l}=\kappa_l\lambda_i,\quad l=0,\dots,q-1.
    \end{align*}
    By the ordering of the eigenvalues and the first inequality in
    \eqref{eq:smooth_gain_design},
    \begin{align*}
        \kappa_1
        <\frac{\lambda_1}{4}\kappa_0^2
        \leq\frac{\lambda_i}{4}\kappa_0^2.
    \end{align*}
    Since $\lambda_i>0$, multiplying the last inequality by
    $4\lambda_i$ yields
    \begin{align*}
        \kappa_0^2\lambda_i^2-4\kappa_1\lambda_i>0.
    \end{align*}
    Using $a_q=1$, $a_{q-1}=\kappa_0\lambda_i$, and
    $a_{q-2}=\kappa_1\lambda_i$, this inequality becomes
    \begin{align*}
        a_{q-1}^2-4a_{q-2}a_q>0.
    \end{align*}
    For $l=1,\dots,q-2$, the remaining gain inequalities give
    \begin{align*}
        a_{q-1-l}^2-4a_{q-2-l}a_{q-l}
        =\lambda_i^2
        (\kappa_l^2-4\kappa_{l-1}\kappa_{l+1})>0.
    \end{align*}
    Therefore,
    $a_m^2-4a_{m-1}a_{m+1}>0$ for every $m=1,\dots,q-1$.
    Lemma~\ref{lemm:kurtz} therefore implies that all roots of $p_{q,i}$ are
    real. We next determine their signs. Since every coefficient $a_m$ is
    positive,
    \begin{align*}
        p_{q,i}(0)=a_0=\kappa_{q-1}\lambda_i>0,
    \end{align*}
    whereas a root at zero would require $p_{q,i}(0)=0$. Therefore, zero is
    not a root. Next, for any $s>0$, each term $a_ms^m$ is positive, so
    \begin{align*}
        p_{q,i}(s)=\sum_{m=0}^qa_ms^m>0.
    \end{align*}
    A positive number $s$ could be a root only if $p_{q,i}(s)=0$; the above
    inequality therefore excludes every $s>0$. Hence, no root belongs to
    $[0,\infty)$. Because all roots are real, every root must lie in
    $(-\infty,0)$ and is therefore strictly negative. As established
    above, the eigenvalues of $\mathcal A_q$ are precisely the roots of
    $p_{q,i}$ over $i=1,\dots,n$. Consequently, all eigenvalues of
    $\mathcal A_q$ are strictly negative. Since $q\in\mathcal Q$ was
    arbitrary, $\mathcal A_q$ is Hurwitz for every $q\in\mathcal Q$.
    
    \emph{(ii)}
    Between two consecutive switching instants, the order is constant and the
    state equation~\eqref{eq:smooth_state_equation} applies.
    Part (i) shows that $\mathcal A_q$ is Hurwitz. Hence, there exist constants
    $M_q>0$ and $\beta_q>0$ such that
    \begin{align*}
        \|\exp(\mathcal A_q\tau)\|
        \leq M_q\exp(-\beta_q\tau),\qquad \tau\geq0.
    \end{align*}
    This exponential bound guarantees that the integral defining $P_q$
    converges and therefore that $P_q$ is a finite symmetric matrix. Consequently, for every
    nonzero vector $y$,
    \begin{align*}
        y^TP_qy
        =\int_0^\infty\|\exp(\mathcal A_q\tau)y\|^2\mathrm d\tau>0,
    \end{align*}
    where strict positivity follows because the integrand is continuous and
    equals $\|y\|^2>0$ at $\tau=0$. Thus, $P_q$ is positive definite.
    Moreover,
    \begin{align*}
        \mathcal A_q^TP_q+P_q\mathcal A_q
        &=\int_0^\infty
        \left[\mathcal A_q^T\exp(\mathcal A_q^T\tau)
        \exp(\mathcal A_q\tau)\right.\\
        &\qquad\left.+\exp(\mathcal A_q^T\tau)
        \exp(\mathcal A_q\tau)\mathcal A_q\right]\mathrm d\tau\\
        &=\int_0^\infty\frac{\mathrm d}{\mathrm d\tau}
        \left[\exp(\mathcal A_q^T\tau)
        \exp(\mathcal A_q\tau)\right]\mathrm d\tau\\
        &=\left[\exp(\mathcal A_q^T\tau)
        \exp(\mathcal A_q\tau)\right]_{0}^{\infty}\\
        &=0-I=-I.
    \end{align*}
    The first equality follows by substituting the integral definition of
    $P_q$ and using the commutation identities
    $\mathcal A_q\exp(\mathcal A_q\tau)
    =\exp(\mathcal A_q\tau)\mathcal A_q$ and their transposes. The product
    rule gives the second equality. Finally,
    $\exp(\mathcal A_q\tau)\to0$ as $\tau\to\infty$ because
    $\mathcal A_q$ is Hurwitz, whereas both exponential matrices equal the
    identity at $\tau=0$.

    Define $V_q(t)\triangleq\eta^T(t)P_q\eta(t)$. Between two consecutive
    switching instants,
    \begin{align}
        \dot V_q(t)
        &=\eta^T(t)(\mathcal A_q^TP_q+P_q\mathcal A_q)\eta(t)
        +2\eta^T(t)P_q\mathcal B_qx_0^{(q)}(t)\nonumber\\
        &=-\|\eta(t)\|^2
        +2\eta^T(t)P_q\mathcal{B}_qx_0^{(q)}(t)\nonumber\\
        &\leq-\|\eta(t)\|^2
        +2\|\eta(t)\|\,\|P_q\mathcal B_q\|\,|x_0^{(q)}(t)|\nonumber\\
        &\leq-\frac{1}{2}\|\eta(t)\|^2
        +2\|P_q\mathcal B_q\|^2|x_0^{(q)}(t)|^2\nonumber\\
        &\leq-aV_q(t)+
        2\max_{p\in\mathcal Q}\|P_p\mathcal B_p\|^2\|w(t)\|^2.
        \label{eq:smooth_constructed_V}
    \end{align}
    The first line follows from
    $\dot\eta(t)=\mathcal A_q\eta(t)+\mathcal B_qx_0^{(q)}(t)$. Substitution of the
    Lyapunov identity $\mathcal A_q^TP_q+P_q\mathcal A_q=-I$ gives the
    second line, and the Cauchy--Schwarz inequality gives the third line.
    Young's inequality
    $2xy\leq x^2/2+2y^2$, with
    $x=\|\eta(t)\|$ and
    $y=\|P_q\mathcal B_q\||x_0^{(q)}(t)|$, gives the fourth line. For the last
    line, the definitions of $w$ and $a$ yield
    \begin{align*}
        |x_0^{(q)}(t)|&\leq\|w(t)\|,\\
        \frac{1}{2}\|\eta(t)\|^2
        &\geq\frac{V_q(t)}{2\lambda_{\max}(P_q)}\geq aV_q(t),
    \end{align*}
    and
    $\|P_q\mathcal B_q\|^2
    \leq\max_{p\in\mathcal Q}\|P_p\mathcal B_p\|^2$.

    At a switching instant $t_s$, Lemma~\ref{lemm:smooth_reset_map} gives
    \begin{align*}
        \eta(t_s^+)
        =\mathcal R_{q(t_s^-)}\eta(t_s^-)
        +\mathcal S_{q(t_s^-)}w(t_s).
    \end{align*}
    Since every switch increases the order by one,
    $q(t_s^+)=q(t_s^-)+1$. Substituting the reset relation into the
    Lyapunov function and using
    $\|x+y\|^2\leq2\|x\|^2+2\|y\|^2$ yield
    \begin{align*}
        &\quad V_{q(t_s^+)}(t_s^+) \\
        &=\left\|P_{q(t_s^+)}^{1/2}\eta(t_s^+)\right\|^2\\
        &\leq2\left\|P_{q(t_s^+)}^{1/2}
        \mathcal R_{q(t_s^-)}\eta(t_s^-)\right\|^2+2\left\|P_{q(t_s^+)}^{1/2}
        \mathcal S_{q(t_s^-)}w(t_s)\right\|^2\\
        &=2\eta(t_s^-)^T\mathcal R_{q(t_s^-)}^T
        P_{q(t_s^+)}\mathcal R_{q(t_s^-)}\eta(t_s^-)+2\left\|P_{q(t_s^+)}^{1/2}
        \mathcal S_{q(t_s^-)}\right\|^2\|w(t_s)\|^2\\
    \end{align*}
    The first term satisfies
    \begin{align*}
        &\qquad 2\eta(t_s^-)^T\mathcal R_{q(t_s^-)}^T
        P_{q(t_s^+)}\mathcal R_{q(t_s^-)}\eta(t_s^-)\\
        &\quad\leq
        2\lambda_{\max}\!\left(
        \mathcal R_{q(t_s^-)}^TP_{q(t_s^+)}
        \mathcal R_{q(t_s^-)}\right)\|\eta(t_s^-)\|^2\\
        &\quad\leq
        \frac{2\lambda_{\max}\!\left(
        \mathcal R_{q(t_s^-)}^TP_{q(t_s^+)}
        \mathcal R_{q(t_s^-)}\right)}
        {\lambda_{\min}(P_{q(t_s^-)})}
        V_{q(t_s^-)}(t_s^-)\\
        &\quad\leq\mu V_{q(t_s^-)}(t_s^-),
    \end{align*}
    where the second inequality uses
    $V_{q(t_s^-)}(t_s^-)
    \geq\lambda_{\min}(P_{q(t_s^-)})\|\eta(t_s^-)\|^2$, and the last one
    follows from the definition of $\mu$. Similarly, the second term is
    bounded by
    \begin{align*}
        2\left\|P_{q(t_s^+)}^{1/2}
        \mathcal S_{q(t_s^-)}\right\|^2\|w(t_s)\|^2\leq2\max_{1\leq q<q_{\max}}
        \left\|P_{q+1}^{1/2}\mathcal S_q\right\|^2\|w(t_s)\|^2.
    \end{align*}
    Combining these two bounds gives
    \begin{align}
        V_{q(t_s^+)}(t_s^+)
        &\leq\mu V_{q(t_s^-)}(t_s^-)\nonumber\\
        &\quad+2\max_{1\leq q<q_{\max}}
        \left\|P_{q+1}^{1/2}\mathcal S_q\right\|^2\|w(t_s)\|^2.
        \label{eq:smooth_V_reset_bound}
    \end{align}

    Consider an interval between two consecutive switching instants $t_s$ and
    $t_{s+1}$, and denote the constant order on this interval by $q$.
    Inequality~\eqref{eq:smooth_constructed_V} can be written as
    \begin{align*}
        \dot V_q(t)+aV_q(t)
        \leq2\max_{p\in\mathcal Q}\|P_p\mathcal B_p\|^2\|w(t)\|^2.
    \end{align*}
    Multiplying both sides by the integrating factor $\exp(at)$ gives
    \begin{align*}
        \frac{\mathrm d}{\mathrm dt}\left[\exp(at)V_q(t)\right]
        \leq2\max_{p\in\mathcal Q}\|P_p\mathcal B_p\|^2
        \exp(at)\|w(t)\|^2.
    \end{align*}
    Integrating from $t_s$ to $t_{s+1}$ yields
    \begin{align*}
        &\exp(at_{s+1})V_q(t_{s+1}^-)
        -\exp(at_s)V_q(t_s^+)\\
        &\quad\leq2\max_{p\in\mathcal Q}\|P_p\mathcal B_p\|^2
        \int_{t_s}^{t_{s+1}}\exp(a\tau)\|w(\tau)\|^2\mathrm d\tau.
    \end{align*}
    Dividing by $\exp(at_{s+1})$ gives
    \begin{align}
        V_q(t_{s+1}^-)
        &\leq\exp[-a(t_{s+1}-t_s)]
        V_q(t_s^+)+2\max_{p\in\mathcal Q}\|P_p\mathcal B_p\|^2\nonumber\\
        &\qquad\times
        \int_{t_s}^{t_{s+1}}\exp[-a(t_{s+1}-\tau)]
        \|w(\tau)\|^2
        \mathrm d\tau.
        \label{eq:smooth_V_interval_bound}
    \end{align}

    Now fix $t>0$ and let $N$ be the number of switches on $(0,t]$. Since
    $q(0)=1$ and every switch increases the order by one,
    \begin{align*}
        0\leq N\leq q_{\max}-1.
    \end{align*}
    We now iterate the continuous-evolution estimate
    \eqref{eq:smooth_V_interval_bound} and the reset estimate
    \eqref{eq:smooth_V_reset_bound}. Suppose first that $N\geq1$, and denote
    the switching instants by $0<t_1<\cdots<t_N\leq t$. When tracking only the
    contribution propagated from $V_{q(0)}(0)$, each application of
    \eqref{eq:smooth_V_reset_bound} contributes a factor $\mu$, whereas each
    application of \eqref{eq:smooth_V_interval_bound} contributes the
    exponential-decay factor for that interval. Repeated application of these
    two inequalities therefore gives
    \begin{align*}
        &\mu^N\exp[-a(t-t_N)]\exp[-a(t_N-t_{N-1})]\cdots\\
        &\quad\times\exp[-a(t_2-t_1)]\exp(-at_1)V_{q(0)}(0)\\
        &=\mu^N\exp(-at)V_{q(0)}(0),
    \end{align*}
    because the interval lengths sum to
    $(t-t_N)+(t_N-t_{N-1})+\cdots+(t_2-t_1)+t_1=t$.
    When $N=0$, there is no reset and the same expression reduces to
    $\exp(-at)V_{q(0)}(0)$.
    A continuous-input contribution generated at time $\tau$ is multiplied
    by the remaining exponential-decay factor
    $\exp[-a(t-\tau)]$ and by at most $N$ reset factors. Hence, the sum of all
    continuous-input contributions is bounded by
    \begin{align*}
        &2\mu^N\max_{p\in\mathcal Q}\|P_p\mathcal B_p\|^2
        \int_0^t\exp[-a(t-\tau)]\|w(\tau)\|^2\mathrm d\tau\\
        &\quad\leq\frac{2\mu^N}{a}
        \max_{p\in\mathcal Q}\|P_p\mathcal B_p\|^2
        \sup_{0\leq\tau\leq t}\|w(\tau)\|^2.
    \end{align*}
    Finally, the input term introduced at the $r$th reset passes through
    $N-r$ later resets. Dropping the subsequent exponential-decay factors,
    which are no greater than one, bounds the total reset contribution by
    \begin{align*}
        2\max_{1\leq q<q_{\max}}
        \left\|P_{q+1}^{1/2}\mathcal S_q\right\|^2
        \sum_{j=0}^{N-1}\mu^j
        \sup_{0\leq\tau\leq t}\|w(\tau)\|^2.
    \end{align*}
    The sum is absent when $N=0$. Adding the three contributions gives
    \begin{align*}
        V_{q(t)}(t)
        &\leq\mu^N\exp(-at)V_{q(0)}(0)\\
        &\quad+\frac{2\mu^N}{a}
        \max_{p\in\mathcal Q}\|P_p\mathcal B_p\|^2
        \sup_{0\leq\tau\leq t}\|w(\tau)\|^2\\
        &\quad+2\max_{1\leq q<q_{\max}}
        \left\|P_{q+1}^{1/2}\mathcal S_q\right\|^2\\
        &\qquad\times
        \sum_{j=0}^{N-1}\mu^j
        \sup_{0\leq\tau\leq t}\|w(\tau)\|^2.
    \end{align*}
    Applying $N\leq q_{\max}-1$ and the definition of $G$ reduces this to
    \begin{align*}
        V_{q(t)}(t)
        &\leq\mu^{q_{\max}-1}\exp(-at)V_{q(0)}(0)\\
        &\quad+G\sup_{0\leq\tau\leq t}\|w(\tau)\|^2.
    \end{align*}
    Since
    \begin{align*}
        \left[\min_{p\in\mathcal Q}\lambda_{\min}(P_p)\right]
        \|\eta(t)\|^2
        \leq V_q(t)
        \leq
        \left[\max_{p\in\mathcal Q}\lambda_{\max}(P_p)\right]
        \|\eta(t)\|^2,
    \end{align*}
    the estimate for $V_{q(t)}(t)$ together with these eigenvalue bounds implies
    \begin{align*}
        \|\eta(t)\|^2
        &\leq C_{\mathrm{ISS}}^2\exp(-at)\|\eta(0)\|^2\\
        &\quad+\gamma_{\mathrm{ISS}}^2
        \sup_{0\leq\tau\leq t}\|w(\tau)\|^2.
    \end{align*}
    Taking square roots and using $\sqrt{x+y}\leq\sqrt x+\sqrt y$ proves
    \eqref{eq:smooth_switched_bound}.

    \emph{(iii)}
    Because $q_k$ is integer-valued, nondecreasing, and bounded above by
    $q_{\max}$, there is an integer $K$ such that
    $q_k=q_\infty$ for every $k\geq K$. Suppose first that $q_\infty<q_{\max}$. No switch occurs after $t_K$.
    Therefore, the first case of~\eqref{eq:smooth_switching_rule} must be
    false, and
    \begin{align}
        r_k>\varepsilon_k,\qquad k\geq K.
        \label{eq:smooth_final_order_reduction}
    \end{align}
    The tracking error is continuous at every decision instant, so
    $J(t_k^-)=J(t_k^+)=J(t_k)$. If $J(t_k)>0$, substituting
    \eqref{eq:smooth_relative_reduction} into
    \eqref{eq:smooth_final_order_reduction} gives
    \begin{align}
        J(t_{k+1})<(1-\varepsilon_k)J(t_k).
        \label{eq:smooth_endpoint_contraction}
    \end{align}
    If $J(t_k)=0$, then $J(t_{k+1})$ must also be zero; otherwise
    $r_k=-\infty$ and the switching law would increase the order. Thus,
    \eqref{eq:smooth_endpoint_contraction} holds with a non-strict inequality
    in all cases. Iterating this inequality and applying the lower estimate in
    Lemma~\ref{lemm:smooth_monitor_equivalence} proves
    \eqref{eq:smooth_unsaturated_error_bound}.

    Since $0<\varepsilon_k<1$,
    $\log(1-\varepsilon_k)\leq-\varepsilon_k$. Hence,
    \begin{align*}
        \prod_{\ell=K}^{k-1}(1-\varepsilon_\ell)
        \leq
        \exp\left(-\sum_{\ell=K}^{k-1}\varepsilon_\ell\right).
    \end{align*}
    The product therefore tends to zero whenever
    $\sum_k\varepsilon_k=\infty$. For the schedule
    \eqref{eq:smooth_threshold_schedule}, this condition holds if
    $\varepsilon_\infty>0$. When $\varepsilon_\infty=0$, it reduces to the
    divergence of the $p$-series
    $\sum_k\varepsilon_0/(k+1)^\alpha$, which occurs exactly for
    $0<\alpha\leq1$. This proves the endpoint convergence assertion.

    Now suppose that $q_\infty=q_{\max}$, and choose $T$ after the last
    switch. Then $q(t)=q_{\max}$ for every $t\geq T$, so
    \eqref{eq:smooth_state_equation} becomes the linear time-invariant system
    \begin{align*}
        \dot\eta(t)=\mathcal A_{q_{\max}}\eta(t)
        +\mathcal B_{q_{\max}}x_0^{(q_{\max})}(t),
        \qquad t\geq T.
    \end{align*}
    Applying the variation-of-constants formula to this system with initial
    state $\eta(T)$ gives
    \begin{align*}
        \eta(t)&=\exp(\mathcal A_{q_{\max}}(t-T))\eta(T)\\
        &\quad+\int_T^t\exp(\mathcal A_{q_{\max}}(t-\tau))
        \mathcal B_{q_{\max}}x_0^{(q_{\max})}(\tau)\mathrm d\tau.
    \end{align*}
    We next convert the asymptotic magnitude of the input
    $x_0^{(q_{\max})}$ into an asymptotic bound on the state. Define
    \begin{align*}
        \bar d\triangleq
        \limsup_{t\to\infty}|x_0^{(q_{\max})}(t)|.
    \end{align*}
    If $\bar d=\infty$, the claimed bound is immediate. Suppose instead that
    $\bar d<\infty$. By the definition of the limit superior, for every
    $\delta>0$ there exists $T_\delta\geq T$ such that
    \begin{align*}
        |x_0^{(q_{\max})}(\tau)|\leq \bar d+\delta,
        \qquad \tau\geq T_\delta.
    \end{align*}
    Applying the variation-of-constants formula from $T_\delta$ yields
    \begin{align*}
        \eta(t)
        &=\exp(\mathcal A_{q_{\max}}(t-T_\delta))\eta(T_\delta)\\
        &\quad+\int_{T_\delta}^t
        \exp(\mathcal A_{q_{\max}}(t-\tau))
        \mathcal B_{q_{\max}}x_0^{(q_{\max})}(\tau)\mathrm d\tau.
    \end{align*}
    Therefore, for $t\geq T_\delta$,
    \begin{align*}
        \|\eta(t)\|
        &\leq
        \|\exp(\mathcal A_{q_{\max}}(t-T_\delta))
        \eta(T_\delta)\|\\
        &\quad+(\bar d+\delta)
        \int_{T_\delta}^t
        \|\exp(\mathcal A_{q_{\max}}(t-\tau))
        \mathcal B_{q_{\max}}\|\mathrm d\tau\\
        &=\|\exp(\mathcal A_{q_{\max}}(t-T_\delta))
        \eta(T_\delta)\|\\
        &\quad+(\bar d+\delta)
        \int_0^{t-T_\delta}
        \|\exp(\mathcal A_{q_{\max}}s)
        \mathcal B_{q_{\max}}\|\mathrm ds\\
        &\leq\|\exp(\mathcal A_{q_{\max}}(t-T_\delta))
        \eta(T_\delta)\|\\
        &\quad+(\bar d+\delta)
        \int_0^\infty
        \|\exp(\mathcal A_{q_{\max}}s)
        \mathcal B_{q_{\max}}\|\mathrm ds,
    \end{align*}
    where the equality uses the change of variable $s=t-\tau$. Because
    $\mathcal A_{q_{\max}}$ is Hurwitz, the first term tends to zero and the
    improper integral is finite. Taking the limit superior gives
    \begin{align*}
        \limsup_{t\to\infty}\|\eta(t)\|
        &\leq(\bar d+\delta)
        \int_0^\infty
        \|\exp(\mathcal A_{q_{\max}}s)\mathcal B_{q_{\max}}\|
        \mathrm ds.
    \end{align*}
    Letting $\delta\downarrow0$, substituting the definition of $\bar d$,
    and using $\|e\|\leq\|\eta\|$ proves the saturated-order
    bound~\eqref{eq:smooth_saturated_bound}.

    Finally, the same argument applies after the order settles at any
    $q_\infty\in\mathcal Q$. In that case, the final subsystem is
    \begin{align*}
        \dot\eta(t)
        =\mathcal A_{q_\infty}\eta(t)
        +\mathcal B_{q_\infty}x_0^{(q_\infty)}(t).
    \end{align*}
    Because $\mathcal A_{q_\infty}$ is Hurwitz, the homogeneous response
    converges to zero. If $x_0^{(q_\infty)}(t)\to0$, split the convolution at
    a sufficiently large time: the contribution over the fixed initial
    interval decays exponentially, while the remaining contribution is
    bounded by
    \begin{align*}
        \sup_{\tau\geq T_\delta}|x_0^{(q_\infty)}(\tau)|
        \int_0^\infty
        \|\exp(\mathcal A_{q_\infty}s)\mathcal B_{q_\infty}\|
        \mathrm ds,
    \end{align*}
    which can be made arbitrarily small. Therefore,
    $\eta(t)\to0$ and hence $e(t)\to0$.
\end{IEEEproof}

\begin{rem}
    Part~(i) shows that every individual order produces a stable continuous
    subsystem. This fact alone would not generally guarantee stability under
    switching, because a reset may increase the Lyapunov function even when
    both the old and new subsystems are stable. Part~(ii) accounts explicitly
    for this increase through the factor $\mu$ and the reset-input matrices
    $\mathcal S_q$. 
\end{rem}

\begin{rem}
    The vector $w(t)$ contains all target derivatives required by the
    admissible controller orders. When order $q$ is active, the continuous
    dynamics~\eqref{eq:smooth_state_equation} are driven only by
    $x_0^{(q)}(t)$, whereas an order increase at $t_s$ introduces the reset
    term $\mathcal S_qw(t_s)$ from Lemma~\ref{lemm:smooth_reset_map}.
    Accordingly, the second term in
    \eqref{eq:smooth_switched_bound} bounds both the continuous forcing and
    the reset-induced state increments. If $w(t)$ is bounded, this term remains
    finite, and hence $\eta(t)$ remains uniformly bounded throughout all
    admissible order transitions. The first term represents the zero-input
    response and decays exponentially from the initial state. Because the
    input term uses the supremum over the entire past, however,
    \eqref{eq:smooth_switched_bound} alone does not imply convergence when
    $w(t)\to0$. Instead, Part~(iii) establishes convergence after the order
    settles: $x_0^{(q_\infty)}(t)\to0$ implies $e(t)\to0$. Finally, the
    constants in \eqref{eq:smooth_switched_bound} account for every
    admissible order and reset, so the estimate provides a uniform robustness bound.
\end{rem}

\begin{rem}
    The condition $q_\infty<q_{\max}$ means that the active order repeatedly
    achieves the relative reduction required by the monitor, so no further
    integral state is activated. It does not, by itself, mean that the target
    has an exact internal model of order $q_\infty$. The monitor compares only
    the two endpoints of each decision interval. Accordingly,
    \eqref{eq:smooth_unsaturated_error_bound} guarantees contraction at the
    decision instants, but it does not exclude a larger error between those
    instants. 
\end{rem}

\begin{rem}
    The equality $q_\infty=q_{\max}$ indicates that the adaptation has reached
    its prescribed complexity limit; it does not indicate loss of stability.
    Bound~\eqref{eq:smooth_saturated_bound} shows that the remaining tracking
    error is governed by two factors: the asymptotic magnitude of the
    unmodeled derivative $x_0^{(q_{\max})}$ and the induced gain of the final
    stable subsystem. In particular, a polynomial target of degree $d$ is
    tracked asymptotically whenever $q_\infty\geq d+1$, since then
    $x_0^{(q_\infty)}=0$. 
\end{rem}

\section{Detection of insufficient orders}
\label{sec:order_detection}


The preceding convergence analysis identifies the admissible final orders,
but it does not yet explain how the switching mechanism recognizes an
insufficient active order or how quickly it does so. For a polynomial target
of degree $d$, the error behavior depends on the degree of under-modeling. If
$q<d$, the unmatched polynomial forcing causes the tracking error to grow
polynomially, whereas the critical order $q=d$ produces a nonzero steady
residual. In either case, the relative improvement $r_k$ in
\eqref{eq:smooth_relative_reduction} eventually falls below the threshold
$\varepsilon_k$ in \eqref{eq:smooth_threshold_schedule}. Therefore, whenever
$q<q_{\max}$, the switching rule~\eqref{eq:smooth_switching_rule} rejects the
active order after finitely many decision intervals. The critical case $q=d$
is the most difficult to detect because its residual remains bounded rather
than growing. We therefore also derive an explicit upper bound on its
rejection time. The following theorem formalizes these claims.

\begin{theor}
    \label{theor:order_detection}
    Suppose that the hypotheses of
    Theorem~\ref{theor:adaptive_dhic} hold, the decision instants are
    uniformly spaced as $t_k=kT_d$, 
    $\varepsilon_\infty>0$, and the target is a polynomial of exact degree
    $d\geq1$, i.e.,
    $x_0(t)=\sum_{r=0}^d a_rt^r$ with $a_d\neq0$.
    Then the following statements hold.
    \begin{enumerate}[(i)]
    \item Every active order $q\leq d$ is rejected after finitely many
    decision intervals unless $q=q_{\max}$. Consequently,
    \begin{align}
        q_\infty\geq\min\{d+1,q_{\max}\}.
        \label{eq:no_underfitting}
    \end{align}

    \item Suppose that order $q=d$ is activated, and let $k_d$ be the first
    decision index such that $q_{k_d}=d$. If $d<q_{\max}$, then its
    rejection time satisfies the explicit bound
    \begin{align*}
        t_{\mathrm{rej},d}
        \leq t_{k_d}+
        \left(\left\lceil\frac{\tau_d}{T_d}\right\rceil+1\right)T_d,
    \end{align*}
    where $\tau_d$ is a constant defined in
    \eqref{eq:critical_detection_time}. In particular, order $d$ is rejected
    no later than one decision interval after the first decision instant not
    smaller than $t_{k_d}+\tau_d$.
    \end{enumerate}
\end{theor}

\begin{IEEEproof}
    \emph{(i)} For $z\in\mathbb R^n$, define
    \begin{align}
        \mathcal J_{\mathrm m}(z)
        &\triangleq\max_{i=1,\dots,n}\rho_i(z),
        \label{eq:monitor_norm}
    \end{align}
    where $\rho_i(z)
        \triangleq
        \left[
        \frac{1}{2}\sum_{j\in\mathcal N_i}(z_i-z_j)^2
        +a_{i0}z_i^2
        \right]^{1/2}$. 
    By definition, $J(t)=\mathcal J_{\mathrm m}(e(t))$. 
    The reverse triangle inequality applied to the local quantity
    $\rho_i$ in \eqref{eq:monitor_norm} gives
    \begin{align*}
        |\rho_i(z)-\rho_i(y)|\leq\rho_i(z-y).
    \end{align*}
    Using the standard inequality
    $|\max_i a_i-\max_i b_i|\leq\max_i|a_i-b_i|$ and then taking the
    maximum of the local bounds over $i$ yields
    \begin{align}
        |\mathcal J_{\mathrm m}(z)-\mathcal J_{\mathrm m}(y)|
        &\leq\max_{i=1,\dots,n}|\rho_i(z)-\rho_i(y)|\nonumber\\
        &\leq\max_{i=1,\dots,n}\rho_i(z-y)
        =\mathcal J_{\mathrm m}(z-y).
        \label{eq:monitor_reverse_triangle}
    \end{align}
    Hence,
    \begin{align*}
        \mathcal J_{\mathrm m}^2(z-y)
        &\leq \sum_{i=1}^n
        \left[
        \frac{1}{2}\sum_{j\in\mathcal N_i}
        \big[(z_i-y_i)-(z_j-y_j)\big]^2
        +a_{i0}(z_i-y_i)^2
        \right]\\
        &=(z-y)^T(L+A_0)(z-y).
    \end{align*}
    By
    Lemma~\ref{lemm:varEig},
    \begin{align*}
        \mathcal J_{\mathrm m}^2(z-y)
        \leq (z-y)^T(L+A_0)(z-y)\leq\lambda_n\|z-y\|^2,
    \end{align*}
    which further gives
    \begin{align*}
        \mathcal J_{\mathrm m}(z-y)
        \leq\sqrt{\lambda_n}\|z-y\|.
    \end{align*}
    Substituting this bound into
    \eqref{eq:monitor_reverse_triangle} yields
    \begin{align}
        |\mathcal J_{\mathrm m}(z)-\mathcal J_{\mathrm m}(y)|
        \leq\sqrt{\lambda_n}\|z-y\|.
        \label{eq:monitor_lipschitz}
    \end{align}

    We next characterize the response of an
    insufficient fixed order. Suppose that order $q\leq d$ is activated and
    never subsequently changed. Its error dynamics are given directly by
    \eqref{eq:smooth_error_dynamics}.
    By Theorem~\ref{theor:adaptive_dhic}, the homogeneous system associated
    with~\eqref{eq:smooth_error_dynamics} is exponentially stable.
    Therefore, the error can be decomposed as
    \begin{align}
        e(t)=e_{\mathrm h}(t)+e_{\mathrm p}(t),
        \label{eq:underorder_solution_decomposition}
    \end{align}
    where $e_{\mathrm h}$ and its derivatives decay exponentially and
    $e_{\mathrm p}$ is a polynomial particular solution.

    Consider first $q<d$ and set $m\triangleq d-q\geq1$. Since the terms
    $a_rt^r$ with $r<q$ vanish after $q$ differentiations, differentiating
    the target polynomial term by term gives
    \begin{align}
        x_0^{(q)}(t)
        &=\sum_{r=q}^d
        a_r\frac{r!}{(r-q)!}t^{r-q}\nonumber\\
        &=a_d\frac{d!}{(d-q)!}t^{d-q}
        +\sum_{r=q}^{d-1}
        a_r\frac{r!}{(r-q)!}t^{r-q}\nonumber\\
        &=\frac{a_dd!}{m!}t^m+O(t^{m-1}).
        \label{eq:target_derivative_leading_term}
    \end{align}
    Since $e_{\mathrm p}(t)\in\mathbb R^n$, write the polynomial particular
    solution as
    \begin{align}
        e_{\mathrm p}(t)=c_mt^m+c_{m-1}t^{m-1}
        +\cdots+c_0,\qquad c_0,\dots,c_m\in\mathbb R^n.
        \label{eq:underorder_particular_polynomial}
    \end{align}
    After substituting $e_{\mathrm p}(t)$ into
    \eqref{eq:smooth_error_dynamics}, the term $e_{\mathrm p}^{(q)}(t)$ and
    the terms
    $\kappa_l(L+A_0)e_{\mathrm p}^{(q-1-l)}(t)$ for
    $l=0,\dots,q-2$ contain at least one derivative of
    $e_{\mathrm p}(t)$ and therefore have degree at most $m-1$. The remaining
    term, corresponding to $l=q-1$, contains no derivative. Consequently, its
    leading component $\kappa_{q-1}(L+A_0)c_mt^m$ is the only term of degree
    $m$ on the left-hand side. By
    \eqref{eq:target_derivative_leading_term}, the coefficient of $t^m$ on
    the right-hand side of \eqref{eq:smooth_error_dynamics} is
    $-a_dd!\mathbf1/m!$. Equality of the two polynomial expressions therefore
    requires their coefficients of $t^m$ to satisfy
    \begin{align}
        \kappa_{q-1}(L+A_0)c_m
        =-\frac{a_dd!}{m!}\mathbf1.
        \label{eq:underorder_leading_balance}
    \end{align}
    Multiplying both sides of \eqref{eq:underorder_leading_balance} by
    $(L+A_0)^{-1}/\kappa_{q-1}$ gives
    \begin{align}
        c_m=v_q\triangleq
        -\frac{a_dd!}{m!\kappa_{q-1}}(L+A_0)^{-1}\mathbf1.
        \label{eq:underorder_leading_coefficient}
    \end{align}
    Since $a_d\neq0$ and $L+A_0$ is nonsingular, $v_q\neq\mathbf{0}$.
    Substituting \eqref{eq:underorder_particular_polynomial} and
    $c_m=v_q$ from \eqref{eq:underorder_leading_coefficient} into the solution
    decomposition~\eqref{eq:underorder_solution_decomposition} gives
    \begin{align*}
        e(t)&=e_{\mathrm h}(t)+v_qt^m+c_{m-1}t^{m-1}
        +\sum_{r=0}^{m-2}c_rt^r.
    \end{align*}
    Here, the last sum is interpreted as zero when $m=1$ and is therefore
    $O(t^{m-2})$ for every $m\geq1$. Moreover, exponential stability implies
    $e_{\mathrm h}(t)=o(t^{m-2})$, because exponential decay dominates every
    polynomial power. Therefore,
    \begin{align}
        e(t)=v_qt^m+c_{m-1}t^{m-1}+O(t^{m-2}).
        \label{eq:underorder_error_second_order}
    \end{align}
    Equation~\eqref{eq:underorder_leading_balance} shows that
    $[(L+A_0)v_q]_i\neq0$ for every $i$. Hence $\rho_i(v_q)>0$; otherwise,
    all the relative and target-access terms defining $\rho_i(v_q)$ would
    vanish, which implies $[(L+A_0)v_q]_i=0$. Define
    \begin{align*}
        a_{q,i}&\triangleq\rho_i(v_q)>0,\\
        b_{q,i}&\triangleq\frac{1}{a_{q,i}}
        \Bigg[\frac{1}{2}\sum_{j\in\mathcal N_i}
        (v_{q,i}-v_{q,j})
        \big((c_{m-1})_i-(c_{m-1})_j\big)\nonumber\\
        &\hspace{42mm}+a_{i0}v_{q,i}(c_{m-1})_i\Bigg].
    \end{align*}
    Substituting \eqref{eq:underorder_error_second_order} into the definition
    of $\rho_i$ in \eqref{eq:monitor_norm} and collecting equal powers of
    $t$ gives
    \begin{align*}
        \rho_i^2(e(t))
        =a_{q,i}^2t^{2m}
        +2a_{q,i}b_{q,i}t^{2m-1}
        +O(t^{2m-2}).
    \end{align*}
    Since $a_{q,i}>0$, taking the square root and using
    $(1+s)^{1/2}=1+s/2+O(s^2)$ as $s\to0$ yields
    \begin{align*}
        \rho_i(e(t))
        &=a_{q,i}t^m
        \left[
        1+\frac{2b_{q,i}}{a_{q,i}t}+O(t^{-2})
        \right]^{1/2}\\
        &=a_{q,i}t^m
        \left[
        1+\frac{b_{q,i}}{a_{q,i}t}+O(t^{-2})
        \right]\\
        &=a_{q,i}t^m+b_{q,i}t^{m-1}+O(t^{m-2}).
    \end{align*}
    Define the largest leading coefficient and the set of indices attaining
    it by
    \begin{align*}
        \alpha_q&\triangleq\max_{i=1,\dots,n}a_{q,i}
        =\mathcal J_{\mathrm m}(v_q)>0,\\
        \mathcal I_q&\triangleq
        \{i\in\{1,\dots,n\}:a_{q,i}=\alpha_q\}.
    \end{align*}
    For every $i\notin\mathcal I_q$, the positive gap
    $\alpha_q-a_{q,i}$ multiplies the dominant term $t^m$; hence such an
    index cannot maximize $\rho_i(e(t))$ for all sufficiently large $t$.
    Thus, for all sufficiently large $t$, only indices in $\mathcal I_q$ can
    attain the maximum. Defining
    $b_q\triangleq\max_{i\in\mathcal I_q}b_{q,i}$, and using the fact that
    $\mathcal I_q$ is finite, we obtain
    \begin{align}
    \label{eq:underorder_monitor_expansion}
        J(t)
        &=\max_{i\in\mathcal I_q}
        \left[
        \alpha_qt^m+b_{q,i}t^{m-1}+O(t^{m-2})
        \right] \nonumber \\
        &=\alpha_qt^m+b_qt^{m-1}+O(t^{m-2}).
    \end{align}
    Factoring the leading term in
    \eqref{eq:underorder_monitor_expansion} gives
    \begin{align*}
        J(t_k)
        &=\alpha_qt_k^m
        \left[1+\frac{b_q}{\alpha_qt_k}+O(t_k^{-2})\right],\\
        J(t_k+T_d)
        &=\alpha_qt_k^m
        \left[1+\frac{mT_d+b_q/\alpha_q}{t_k}
        +O(t_k^{-2})\right].
    \end{align*}
    Substituting these two expressions into the relative-reduction test
    \eqref{eq:smooth_relative_reduction} and using
    $(1+s)^{-1}=1-s+O(s^2)$ as $s\to0$ yields
    \begin{align}
        r_k
        &=1-\frac{J(t_k+T_d)}{J(t_k)}\nonumber\\
        &=1-
        \frac{1+(mT_d+b_q/\alpha_q)/t_k+O(t_k^{-2})}
        {1+b_q/(\alpha_qt_k)+O(t_k^{-2})}\nonumber\\
        &=1-
        \left[1+\frac{mT_d+b_q/\alpha_q}{t_k}+O(t_k^{-2})\right]
        \left[1-\frac{b_q}{\alpha_qt_k}+O(t_k^{-2})\right]
        \nonumber\\
        &=-\frac{mT_d}{t_k}+O(t_k^{-2}).
        \label{eq:growing_residual_reduction}
    \end{align}
    Thus, $r_k<0$ for all sufficiently large $k$. Since
    $\varepsilon_k>0$, an order $q<d$ cannot remain active indefinitely
    unless it has reached the prescribed maximum order.

    Now consider the critical under-order $q=d$. In this case,
    $x_0^{(d)}(t)=a_dd!$ is constant. Setting $q=d$ in
    \eqref{eq:smooth_error_dynamics} gives
    \begin{align*}
        e^{(d)}(t)
        +\sum_{l=0}^{d-1}\kappa_l(L+A_0)e^{(d-1-l)}(t)
        =-a_dd!\mathbf1.
    \end{align*}
    For a constant particular solution $e(t)=e_d^\star$, all positive-order
    derivatives of $e_d^\star$ vanish. Hence every term in the sum with
    $l<d-1$ is zero, while the term with $l=d-1$ contains
    $e^{(0)}(t)=e_d^\star$. The remaining algebraic equation is therefore
    \begin{align*}
        \kappa_{d-1}(L+A_0)e_d^\star=-a_dd!\mathbf1,
    \end{align*}
    and hence
    \begin{align}
        e_d^\star
        =-\frac{a_dd!}{\kappa_{d-1}}(L+A_0)^{-1}\mathbf1.
        \label{eq:critical_residual}
    \end{align}
    Since $a_d\neq0$ and $L+A_0$ is nonsingular,
    $e_d^\star\neq\mathbf{0}$. Because $\mathcal J_{\mathrm m}$ is a norm,
    the corresponding steady monitor value
    $J_d^\star\triangleq\mathcal J_{\mathrm m}(e_d^\star)$ is strictly
    positive. If order $d$ were to remain active indefinitely, the solution
    would be the sum of $e_d^\star$ and an exponentially decaying homogeneous
    response. Consequently,
    \begin{align}
        e(t)\to e_d^\star,\qquad
        J(t)\to J_d^\star>0.
        \label{eq:critical_monitor_limit}
    \end{align}
    Thus, the two endpoint values used to compute $r_k$ converge to the same
    positive limit, which implies $r_k\to0$. Since
    $\varepsilon_k\geq\varepsilon_\infty>0$, the inequality
    $r_k\leq\varepsilon_k$ must hold after finitely many decision intervals.
    The switching rule then rejects order $d$ whenever $d<q_{\max}$,
    contradicting indefinite activation. Applying this conclusion together
    with the result for $q<d$ successively to every active order
    $q\leq d$ proves~\eqref{eq:no_underfitting}.

    \emph{(ii)} It remains to establish the stated quantitative
    rejection-time bound for the critical under-order. By definition,
    $t_{k_d}$ is the time at which order $d$ becomes active; in
    particular, $k_d=0$ when $d=1$. While order $d$ is active, the
    derivative-state vector is
    \begin{align*}
        \eta(t)=
        \begin{bmatrix}
            e^T(t)&[e^{(1)}(t)]^T&\cdots&[e^{(d-1)}(t)]^T
        \end{bmatrix}^T.
    \end{align*}
    At the constant particular solution $e(t)=e_d^\star$, all derivatives
    $e^{(r)}(t)$, $r=1,\dots,d-1$, are zero. The corresponding equilibrium
    of the derivative-state dynamics is therefore
    \begin{align*}
        \eta_d^\star
        \triangleq
        \begin{bmatrix}
            (e_d^\star)^T&\mathbf{0}^T&\cdots&\mathbf{0}^T
        \end{bmatrix}^T.
    \end{align*}
    Equation~\eqref{eq:critical_residual} ensures that
    $\mathcal A_d\eta_d^\star+\mathcal B_da_dd!=\mathbf{0}$, so
    $\eta_d^\star$ is indeed an equilibrium of
    \eqref{eq:smooth_state_equation} for the constant input
    $x_0^{(d)}=a_dd!$.
    Choose $M_d\geq1$ and $\beta_d>0$ such that
    \begin{align*}
        \|\exp(\mathcal A_dt)\|
        \leq M_d\exp(-\beta_dt),\qquad t\geq0,
    \end{align*}
    which is possible because $\mathcal A_d$ is Hurwitz. Subtracting
    $\eta_d^\star$ from the state equation gives the exact homogeneous
    relation
    \begin{align*}
        \eta(t)-\eta_d^\star
        =\exp[\mathcal A_d(t-t_{k_d})]
        [\eta(t_{k_d})-\eta_d^\star].
    \end{align*}
    It follows from the definition of $M_d$ and $\beta_d$ that
    \begin{align}
        \|\eta(t)-\eta_d^\star\|
        \leq
        M_d\exp[-\beta_d(t-t_{k_d})]
        \|\eta(t_{k_d})-\eta_d^\star\|.
        \label{eq:critical_state_decay}
    \end{align}
    Combining
    \eqref{eq:critical_state_decay} with
    \eqref{eq:monitor_lipschitz} gives
    \begin{align}
        |J(t)-J_d^\star|
        &\leq
        \sqrt{\lambda_n}M_d
        \exp[-\beta_d(t-t_{k_d})]
        \|\eta(t_{k_d})-\eta_d^\star\|.
        \label{eq:critical_monitor_decay}
    \end{align}
    Define
    \begin{align}
        \tau_d
        \triangleq\frac{1}{\beta_d}
        \left[
        \log
        \frac{
        (2+\varepsilon_\infty)
        \sqrt{\lambda_n}M_d
        \|\eta(t_{k_d})-\eta_d^\star\|}
        {\varepsilon_\infty J_d^\star}
        \right]_+,
        \label{eq:critical_detection_time}
    \end{align}
    Define
    \begin{align*}
        \theta_d\triangleq
        \frac{\varepsilon_\infty}{2+\varepsilon_\infty}.
    \end{align*}
    For every $t\geq t_{k_d}+\tau_d$, we have
    \begin{align*}
        &\sqrt{\lambda_n}M_d
        \exp[-\beta_d(t-t_{k_d})]
        \|\eta(t_{k_d})-\eta_d^\star\|\\
        &\quad\leq
        \sqrt{\lambda_n}M_d
        \exp(-\beta_d\tau_d)
        \|\eta(t_{k_d})-\eta_d^\star\|.
    \end{align*}
    If
    \begin{align*}
        \sqrt{\lambda_n}M_d
        \|\eta(t_{k_d})-\eta_d^\star\|
        \leq\theta_dJ_d^\star,
    \end{align*}
    then the logarithm in~\eqref{eq:critical_detection_time} is nonpositive,
    so $\tau_d=0$ and
    \begin{align*}
        &\sqrt{\lambda_n}M_d
        \exp(-\beta_d\tau_d)
        \|\eta(t_{k_d})-\eta_d^\star\|\\
        &\qquad=
        \sqrt{\lambda_n}M_d
        \|\eta(t_{k_d})-\eta_d^\star\|
        \leq\theta_dJ_d^\star.
    \end{align*}
    Otherwise,
    \begin{align*}
        \tau_d
        =\frac{1}{\beta_d}
        \log\frac{\sqrt{\lambda_n}M_d
        \|\eta(t_{k_d})-\eta_d^\star\|}
        {\theta_dJ_d^\star},
    \end{align*}
    and direct substitution gives
    \begin{align*}
        \sqrt{\lambda_n}M_d
        \exp(-\beta_d\tau_d)
        \|\eta(t_{k_d})-\eta_d^\star\|=\theta_dJ_d^\star.
    \end{align*}
    Hence, in either case, the exponential bound is no greater than
    $\theta_dJ_d^\star$.
    Therefore, \eqref{eq:critical_monitor_decay} yields
    \begin{align}
        |J(t)-J_d^\star|
        \leq\theta_dJ_d^\star,
        \qquad t\geq t_{k_d}+\tau_d.
        \label{eq:critical_monitor_relative_bound}
    \end{align}

    At $t=t_k$, \eqref{eq:critical_monitor_relative_bound} gives
    \begin{align*}
        J(t_k)
        \geq J_d^\star-|J(t_k)-J_d^\star|\geq(1-\theta_d)J_d^\star.
    \end{align*}
    Applying the same estimate at $t=t_{k+1}$ and using the triangle
    inequality yields
    \begin{align*}
        J(t_k)-J(t_{k+1})
        &=J(t_k)-J_d^\star
        +J_d^\star-J(t_{k+1})\\
        &\leq |J(t_k)-J_d^\star|
        +|J(t_{k+1})-J_d^\star|\\
        &\leq2\theta_dJ_d^\star.
    \end{align*}
    Therefore,
    \begin{align*}
        r_k
        \leq\frac{2\theta_d}{1-\theta_d}
        =\varepsilon_\infty
        \leq\varepsilon_k,
    \end{align*}
    and the switching rule rejects order $d$ at $t_{k+1}$. The first
    decision instant not smaller than $t_{k_d}+\tau_d$ is at most
    $t_{k_d}+\lceil\tau_d/T_d\rceil T_d$. Therefore, the rejection time
    $t_{\mathrm{rej},d}$ satisfies
    \begin{align}
        t_{\mathrm{rej},d}
        \leq
        t_{k_d}+
        \left(\left\lceil\frac{\tau_d}{T_d}\right\rceil+1\right)T_d.
        \label{eq:critical_switch_bound}
    \end{align}
\end{IEEEproof}

    Theorem~\ref{theor:order_detection} formalizes the distinction between
    OADIC and a fixed-order controller. A fixed order $q\leq d$ permanently
    accepts the residual generated by the unmodeled target derivative,
    whereas OADIC rejects that order in finite time. The critical order
    $q=d$ is the slowest deficiency to detect: its error remains bounded and
    converges to the nonzero value~\eqref{eq:critical_residual}. The bound
    \eqref{eq:critical_switch_bound} shows that weak connectivity, slow modal
    decay, a small threshold, or a large reset transient can all delay the
    detection. 

\begin{figure*}[!t]
    \centering
    \includegraphics[width=\textwidth]{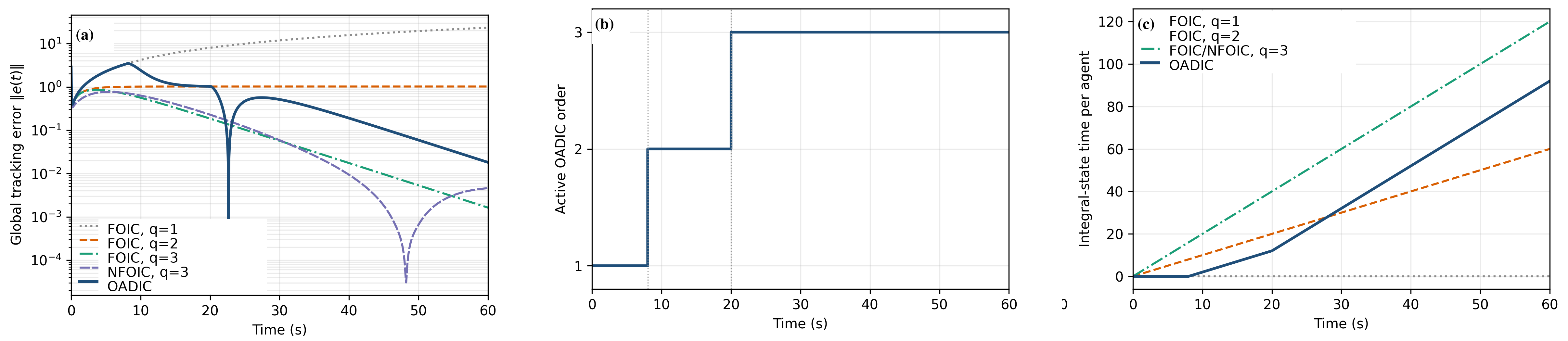}
    \caption{Ten-agent comparison for a quadratic target: (a) global tracking
    errors; (b) order selected by OADIC; and (c) cumulative integral-state
    time per agent.}
    \label{fig:oadic_comparison}
\end{figure*}

\section{Simulation}

In this section, we compare OADIC with both fixed-order linear controllers
and a nonlinear controller. We consider ten agents connected by an
undirected ring augmented with the links $\{1,6\}$ and $\{3,8\}$; hence,
\begin{align*}
    \mathcal E={}&\{\{1,2\},\{2,3\},\ldots,\{9,10\},\{10,1\},\{1,6\},\{3,8\}\}.
\end{align*}
Agents 1, 5, and 8 have access to the target. We choose $\bar d_i=i-5.5$ for $i=1,\ldots,10,$ and $x_0(t)=0.12t^2+0.35t$.
We then set $\bar\zeta_{ij}=\bar d_i-\bar d_j$ on every communication edge
and $\bar\zeta_{i0}=\bar d_i$ for each target-informed agent. Thus, the
desired displacements satisfy the consistency condition in
Theorem~\ref{theor:consistency}. 
We set $q_{\max}=3$ and use the nested gain sequence $\kappa_0=8$, $\kappa_1=3$, $\kappa_2=0.25.$
For this network, $\lambda_1=0.2455$. Therefore,
$\kappa_1<\lambda_1\kappa_0^2/4=3.9285$ and
$\kappa_2<\kappa_1^2/(4\kappa_0)=0.2813$, so the gain conditions in
\eqref{eq:smooth_gain_design} hold. OADIC uses $T_d=4$, $\varepsilon_0=0.22$,
    $\varepsilon_\infty=0.03$, and $\alpha=0.6.$
As linear benchmarks, we use fixed-order integral control (FOIC) with
$q=1,2,3$ and the same gains as OADIC.

We also define a nonlinear-enhanced FOIC (NFOIC) by augmenting the full
order-three FOIC with the fixed-time-type
signed-power feedback commonly used in nonlinear and distributed finite-time
control~\cite{polyakov2012fixed,parsegov2013fixed}:
\begin{align}
    u_i^{\mathrm{NL}}
    ={}&-\sum_{l=0}^{2}\kappa_l\sigma_{il}
    -\gamma_1\operatorname{sig}^{\mu}(\sigma_{i0})
    -\gamma_2\operatorname{sig}^{\nu}(\sigma_{i0}),
    \label{eq:simulation_nonlinear_controller}
\end{align}
where $\operatorname{sig}^{p}(z)\triangleq
|z|^p\operatorname{sign}(z)$, $\mu=0.6$, $\nu=1.4$,
$\gamma_1=1.2$, and $\gamma_2=0.12$. Thus, NFOIC receives
the complete order-three internal model from the outset and further
strengthens its transient feedback through one sublinear and one superlinear
power. This choice favors NFOIC in tracking performance
and allows us to compare that performance with the structural simplicity and
online state allocation of OADIC. We simulate every controller over
$T=60$ using a fixed-step fourth-order Runge--Kutta method with a step size
of $0.002$.

To quantify controller complexity, define the average number of active
integral states per agent over the simulation horizon $[0,T]$ as
\begin{align}
    \bar c(T)\triangleq
    \frac{1}{T}\int_0^T[q(t)-1]\mathrm dt.
    \label{eq:simulation_average_complexity}
\end{align}
Thus, the fixed orders $q=1,2,3$ have $\bar c=0,1,2$, respectively, and
NFOIC has $\bar c=2$. We also calculate the tail
root-mean-square error
\begin{align}
    E_{\mathrm{tail}}
    \triangleq
    \left[
    \frac{1}{0.2T}\int_{0.8T}^{T}\|e(t)\|^2\mathrm dt
    \right]^{1/2}.
    \label{eq:simulation_tail_rms}
\end{align}

\begin{table}[!t]
    \caption{Tracking--complexity--structure comparison over $T=60$ s}
    \label{tab:oadic_comparison}
    \centering
    \begin{tabular}{lccc}
        \hline
        Controller & $E_{\mathrm{tail}}$ & $\bar c(T)$ & Feedback\\
        \hline
        FOIC, $q=1$ & $21.0822$ & $0$ & Linear\\
        FOIC, $q=2$ & $1.0211$ & $1$ & Linear\\
        FOIC, $q=3$ & $0.0039$ & $2$ & Linear\\
        NFOIC, $q=3$ & $0.0029$ & $2$ & Signed power\\
        OADIC & $0.0431$ & $1.5333$ & Linear by order\\
        \hline
    \end{tabular}
\end{table}

Figure~\ref{fig:oadic_comparison}(a) first confirms the role of controller
order. The proportional controller cannot follow the accelerating target,
whereas the order-two controller removes the velocity-dependent component
but retains an acceleration-induced residual. By contrast, both order-three
controllers contain the required internal model and drive the error toward
zero. The additional signed-power terms give NFOIC the
smallest tail error, but they do not reduce its prescribed order or its number
of dynamic states.

As shown in Fig.~\ref{fig:oadic_comparison}(b), OADIC initially retains
proportional feedback, activates the first integral state at $t=8$, and
activates the second at $t=20$. Thereafter, it has the same order-three
internal model as the full-order FOIC and NFOIC, and its error converges
toward zero. Because OADIC delays both activations, its finite-horizon tail
error is larger than those of the controllers that use the complete
order-three structure from $t=0$. Thus, OADIC trades some transient speed for
online complexity allocation.

The benefit of this tradeoff appears in Fig.~\ref{fig:oadic_comparison}(c)
and Table~\ref{tab:oadic_comparison}. OADIC uses no integral state over
$[0,8)$, one per agent over $[8,20)$, and two thereafter. Consequently,
$\bar c(60)=(0\times8+1\times12+2\times40)/60\approx1.5333$. Relative to
either full-order benchmark, the corresponding reduction is
$(2-1.5333)/2\times100\%=23.3\%$. Likewise, the tail RMS error reduction
relative to the fixed order-two controller is
$(1.0211-0.0431)/1.0211\times100\%=95.8\%$.
Therefore, nonlinear enhancement can improve the transient when the complete
internal model is already available, but it does so while retaining all
high-order states and introducing fractional-power computations. OADIC takes
a different route: it keeps every active controller linear and approaches the
required high-order accuracy while activating the corresponding states only
when the observed tracking performance demands them.

\section{Conclusion}

We developed OADIC for distributed dynamic coordination when the required
internal-model order is unknown. Instead of fixing a conservative high order
or introducing nonlinear feedback, we start with proportional feedback and
add integral states only when the observed tracking progress is insufficient.
We thereby move among nested linear controllers according to actual need
while keeping the control law purely linear at every active order. Moreover,
we preserve the existing states and gains during each order transition and
initialize the new state to maintain continuity of the control input.

We also established the main guarantees behind this adaptation mechanism.
First, we derived a cycle-consistency condition that determines whether the
prescribed relative displacements define a feasible coordination pattern.
Next, we constructed a nested gain condition that stabilizes every admissible
fixed-order subsystem. We then established a uniform input-to-state stability
bound that accounts for both target motion and order transitions. Furthermore,
we proved that OADIC rejects every insufficient
nonsaturated order after finitely many decision intervals, even when the
corresponding fixed-order subsystem is stable but leaves a nonzero tracking
error.

Finally, we used numerical comparisons to illustrate the resulting
accuracy--complexity tradeoff. Compared with the fixed order-three controller,
we reduced integral-state time by $23.3\%$ while approaching its tracking
accuracy. Compared with the fixed order-two controller, we reduced the tail
RMS error by approximately $95.8\%$. NFOIC
achieved a smaller finite-horizon tail error, but it retained both integral
states throughout the operation and introduced signed-power feedback at every
agent. Therefore, these results do not suggest that OADIC uniformly
outperforms full-order controllers in transient speed. Instead, they show
that we can obtain the necessary tracking capability with purely linear
active controllers and without carrying the full order throughout the entire
operation. In future work, we will study robustness to noisy measurements and
changing network topologies.

\bibliographystyle{IEEEtran}
\bibliography{refs}

\begin{IEEEbiography}[{\includegraphics[width=1in,height=1.25in,clip,
keepaspectratio]{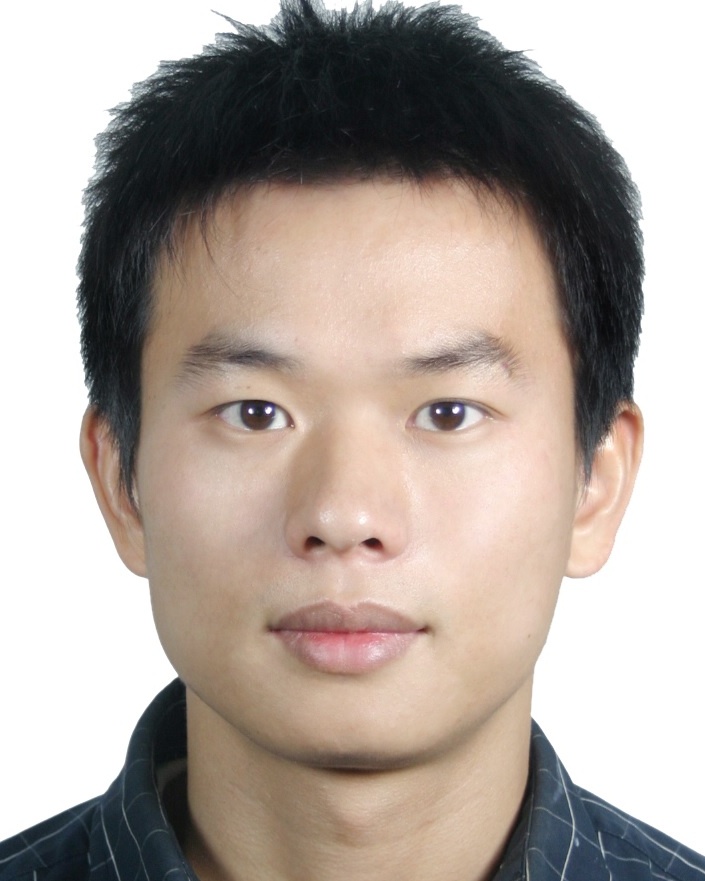}}]{Fei Chen}
(M'12--SM'15) received the Ph.D. degree in control theory and control
engineering from Nankai University, Tianjin, China, in 2009. He is currently
a Professor and Associate Dean of the College of Artificial Intelligence,
Nankai University. From 2009 to 2010, he was a Postdoctoral Researcher with
Utah State University, Logan, UT, USA. He subsequently served as an Associate
Professor at Xiamen University, Xiamen, China, and a Professor at Northeastern
University, Shenyang, China. He has also held visiting appointments at the
University of California, Riverside, CA, USA, from 2017 to 2018, and at the
City University of Hong Kong, Hong Kong, China, in 2008, 2013, 2016, and
2024. Dr. Chen is the author of \emph{Distributed Average Tracking in Multi-Agent
Systems} (Springer, 2020) and has published more than 100 peer-reviewed papers
in international journals and conference proceedings. His research interests include
distributed coopearative control, distributed optimization, and distributed learning.
\end{IEEEbiography}

\end{document}